\documentclass[lettersize,journal]{IEEEtran}
\usepackage{graphicx}
\usepackage{graphics}
\usepackage{subfigure}
\usepackage{amsmath,amsfonts,amsthm}
\usepackage{algorithmic}
\usepackage{algorithm}
\usepackage{array}
\usepackage[caption=false,font=normalsize,labelfont=sf,textfont=sf]{subfig}
\usepackage{textcomp}
\usepackage{stfloats}
\usepackage{url}
\usepackage{verbatim}
\usepackage{color}
\usepackage{cite}
\usepackage{makecell}
\newtheorem{theorem}{Theorem}

\ifodd 1
\newcommand{\com}[1]{\textbf{\color{red} (COMMENT: #1)}} 
\newcommand{\comg}[1]{\textbf{\color{green} (COMMENT: #1)}}
\newcommand{\response}[1]{\textbf{\color{magenta} (RESPONSE: #1)}} 
\else

\newcommand{\com}[1]{}
\newcommand{\comg}[1]{}
\newcommand{\response}[1]{}
\fi

\begin{document}

\title{GNN at the Edge: Cost-Efficient Graph Neural Network Processing over Distributed Edge Servers}

\author{Liekang Zeng,
Chongyu Yang, 
Peng Huang,
Zhi Zhou,
Shuai Yu,
and Xu Chen
\thanks{The authors are with the School of Computer Science and Engineering, Sun Yat-sen University, Guangzhou 510006, China (e-mail: \{zenglk3, yangchy37, huangp57\}@mail2.sysu.edu.cn, \{zhouzhi9, yushuai, chenxu35\}@mail.sysu.edu.cn).}
}



\maketitle

\begin{abstract}
Edge intelligence has arisen as a promising computing paradigm for supporting miscellaneous smart applications that rely on machine learning techniques.
While the community has extensively investigated multi-tier edge deployment for traditional deep learning models (\textit{e.g.} CNNs, RNNs), the emerging Graph Neural Networks (GNNs) are still under exploration, presenting a stark disparity to its broad edge adoptions such as traffic flow forecasting and location-based social recommendation.
To bridge this gap, this paper formally studies the cost optimization for distributed GNN processing over a multi-tier heterogeneous edge network.
We build a comprehensive modeling framework that can capture a variety of different cost factors, based on which we formulate a cost-efficient graph layout optimization problem that is proved to be NP-hard.
Instead of trivially applying traditional data placement wisdom, we theoretically reveal the structural property of quadratic submodularity implicated in GNN's unique computing pattern, which motivates our design of an efficient iterative solution exploiting graph cuts.
Rigorous analysis shows that it provides parameterized constant approximation ratio, guaranteed convergence, and exact feasibility.
To tackle potential graph topological evolution in GNN processing, we further devise an incremental update strategy and an adaptive scheduling algorithm for lightweight dynamic layout optimization.
Evaluations with real-world datasets and various GNN benchmarks demonstrate that our approach achieves superior performance over \textit{de facto} baselines with more than 95.8\% cost reduction in a fast convergence speed.
\end{abstract}


\section{Introduction}

\IEEEPARstart{R}{ecent} advances in neural networks have emerged as a powerful technique in a broad range of smart applications, including object detection, speech recognition, and data mining \cite{lecun2015deep, pouyanfar2018survey, zhou2019edge}.
With the rapidly increasing computation resources as supporting platforms and tremendous training data as materials, Deep Neural Networks (DNNs) are able to extract latent representations from existing statistics, especially data in Euclidean or sequential form \cite{lecun2015deep}.
Yet there are still enormous data generated with complex graph structures including social graph \cite{zhong2020hybrid}, wireless sensor network \cite{marche2020exploit}, \textit{etc.}
Analytics on these graphs bring significant challenges to traditional DNNs due to their structure irregularity and feature complexity \cite{yan2020hygcn}.

To cope with representation learning in the graph domain, Graph Neural Networks (GNNs) have arisen and are attracting growing attention from both research \cite{kipf2016semi, velivckovic2017graph, hamilton2017inductive} and industry communities \cite{zhou2018graph, zhu2019aligraph, ma2019neugraph}.
GNNs integrate graph embedding techniques with convolution to collectively aggregate information from nodes and their dependencies, enabling capturing hierarchical patterns from subgraphs of variable sizes.
Benefited from such advanced ability in modeling graph
structures, GNNs have been recently employed in miscellaneous graph-based analytics in edge computing scenarios.
For example, scoring Point-of-Interest (PoI) for location-based recommendation applications \cite{zhong2020hybrid, chang2020learning}, predicting power outages over power grids \cite{owerko2020optimal, ringsquandl2021power}, forecasting network traffic based on the highway sensors \cite{yu2017spatio, wang2020traffic}, and optimizing resource management in wireless IoT networks \cite{chen2021gnn, he2021overview}.

\begin{figure}[t]
  \centering
  \includegraphics[width=0.9\linewidth]{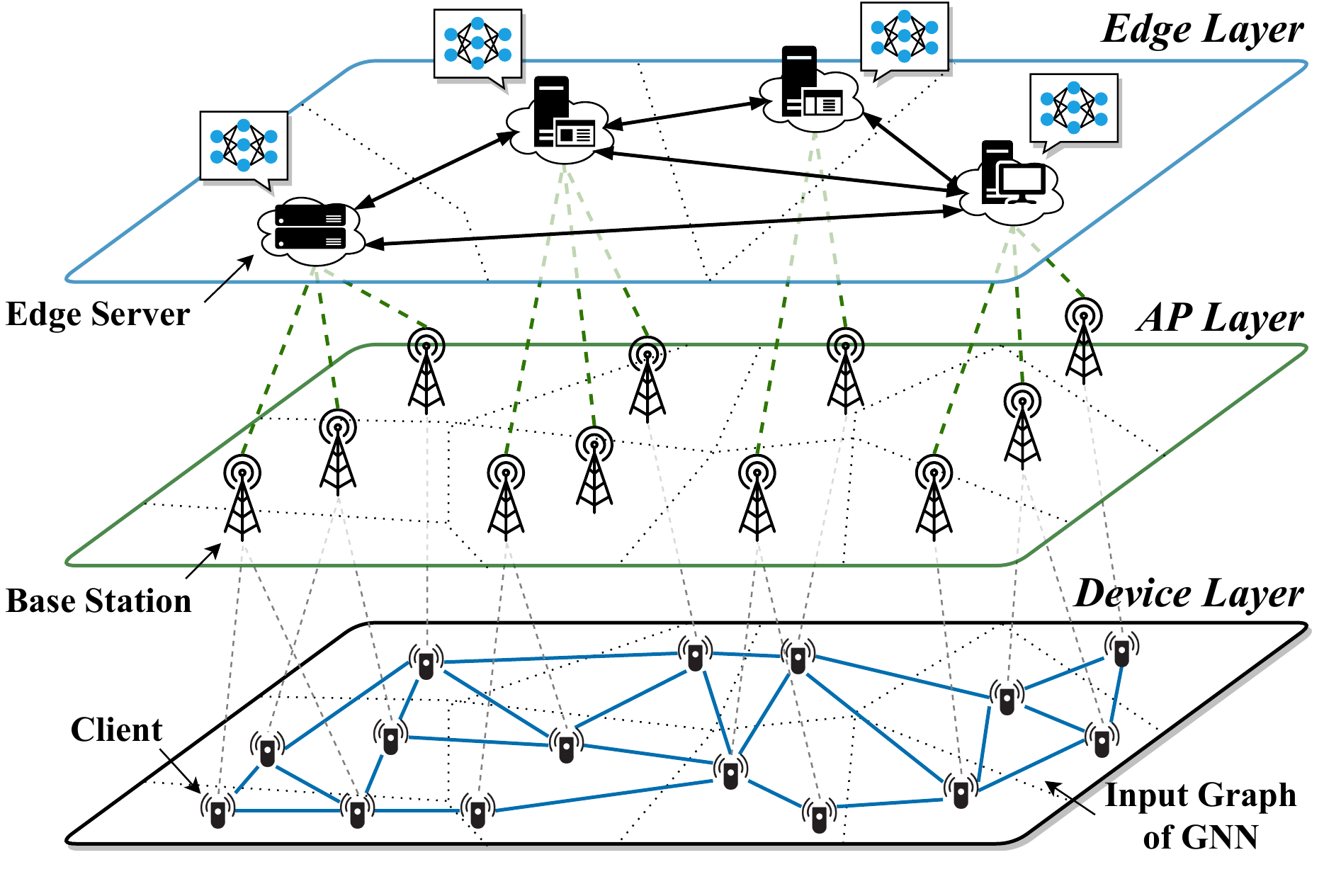}
  \caption{A multi-tier edge network architecture comprising device layer, access point (AP) layer, and edge layer.
  The clients with connections form a data graph (\textit{e.g.}, social web), while each node in the edge network covers a range of users. 
  The edge servers perform distributed GNN processing by parallelizing model execution and exchanging data mutually, in order to infer targeted graph properties (\textit{e.g.}, predicting potential social relationships).
  }
  \label{fig:scenario}
\end{figure}

While traditional Deep Learning (DL) models (\textit{e.g.} CNNs, RNNs) accept inputs from each individual client independently, processing GNNs at the network edge can span over distributed edge servers geographically, due to the dispersed nature of its inputs (graph data).
Fig. \ref{fig:scenario} illustrates how this works, where device clients with connections (\textit{e.g.} social relationship) form a data graph, and edge servers organize an edge network with each one associates a pool of access points (\textit{e.g.} 5G base stations, IoT gateways) for data collection.
When a GNN inference query is raised, each edge server first aggregates a subset of the graph data via APs in certain areas, and next launches a distributed runtime to compute the embeddings through the given GNN model.
During the runtime, the edge servers are orchestrated in a collaborative manner, where they are committed to exchanging necessary graph data with each other, parallelizing GNN execution individually, and repeating this routine iteratively.

Unfortunately, efficient distributed GNN processing is much less understood in the context of edge intelligence.
Existing efforts on edge-enabled distributed intelligence have relied on either offloading computation with external infrastructures \cite{kang2017neurosurgeon, teerapittayanon2017distributed, li2019edge} or partitioning model execution across edge servers \cite{mao2017modnn, hadidi2018distributed, zeng2020coedge}.
Both of them, however, implicitly assume single-point input and an independent serving style between one client and one edge server, which fits CNNs and RNNs but is inapplicable for GNNs.
Furthermore, while some works \cite{ma2019neugraph, yang2019aligraph, zhang2020agl} have considered parallel implementations of GNN, little has been done to investigate how graph layout impacts system performance, neither considering distinct characteristics of edge computing like data uploading and resources heterogeneity \cite{shi2016edge}.
Overall, the absence of distributed GNN processing optimization presents a stark disparity to its increasingly popular applications in edge scenarios.

To bridge the gap, in this paper, we formally study the system cost optimization problem for Distributed GNN Processing over heterogeneous Edge servers (DGPE) by building a novel modeling framework that generalizes to a wide variety of cost factors.
In particular, we formulate the hierarchical architecture of multi-tier edge networks, and characterize the system costs in terms of the life cycle of a DGPE task: data collection, GNN computation, cross-edge data exchanging, and edge server maintenance.
While we strive to coordinate a graph layout for dispatching the GNN's input data to proper edge servers, optimizing the system cost is non-trivial given the intrinsic nexus of the objectives:
1) Edge servers can be heterogeneous regarding the cost of GNN computation and edge server maintenance, for which achieving the best economic outcome will push clients to offload their data to servers with lower costs.
2) As the unique computing pattern of GNN requires frequent data exchange between edge servers for DGPE, to reduce the cost of runtime give priority to placing clients' data to edges in closer proximity.
3) While the input data graph could evolve over time, the system performance may fluctuate and reach sub-optimum even with an optimal static placement decision, which calls for efficient adaptive dynamic optimization.
We highlight that the second challenge is unique to DGPE as opposed to other DL models, and all these intrinsically intertwined yet potentially conflicting goals further complicate the problem.

To address these challenges, we propose GLAD, a novel \textbf{G}raph \textbf{LA}yout sche\textbf{D}uling solution framework that can optimize the system cost for different GNN inference scenarios.
Concretely, we first consider the scenarios where the input graph of GNN is topologically static, such as traffic sensory networks in which the IoT sensors are fixed at dedicated locations.
In this setting, we observe that the cost objective function has properties of pseudo-boolean quadratic and submodularity, inherently implicated by the GNN's unique computing pattern, which motivates us to represent the system cost as a series of decoupled graphs.
Specifically, by constructing a mapping from graph cuts to placement decisions, the cost optimization problem is resolved and addressed by settling a series of minimum $s-t$ graph cuts.
Theoretical performance analysis is rigorously conducted regarding its feasibility, approximation ratio, convergence, and time complexity.
We next consider the scenarios with evolving input graphs, where the input data graph may change topologically over time (\textit{e.g.}, social network).
For this case, we devise an incremental update strategy upon the insights in the static counterpart, and further design an adaptive scheduling algorithm to strike a balance between operating overhead and benefits.
Extensive evaluations against real-world data traces and several GNN benchmarks demonstrate the superiority of GLAD over \textit{de facto} baselines, as well as its convergence, scalability, and adaptability in online settings.

The key contributions are summarized below:

\begin{itemize}
    \item We build a novel modeling framework for system cost optimization on distributed GNN processing over heterogeneous edge servers. Combining the unique characteristic of GNN computation and edge computing environment, a cost-efficient graph layout optimization problem is formulated, which is proved to be NP-hard.
    \item We theoretically reveal that the optimization problem's objective function is pseudo-boolean quadratic and submodular, based on which we propose an efficient algorithm leveraging graph-cuts techniques. We provide a rigorous analysis of its approximation ratio, convergence, and time complexity.
    \item We develop an incremental graph layout improvement strategy to address the potential dynamic evolution of GNN's input data graph, and further design an adaptive scheduling algorithm to well balance the tradeoff between graph layout update overhead and system performance.
    \item We carry out extensive evaluations using realistic datasets and various GNN benchmarks, demonstrating that our approach outperforms existing baselines by up to 95.8\% cost reduction, while converges fast, scales to a larger volume, and adapts to dynamics.
\end{itemize}

The rest of this paper is organized as follows.
Sec. \ref{sec:background} briefly reviews GNNs and their applications at the edge, as well as related work.
Sec. \ref{sec:system_model} introduces the system model for DGPE, with detailed formulation on cost factors.
Sec. \ref{sec:offline_algorithm} presents the motivation, design, and analysis of the algorithm for static input graphs, while Sec. \ref{sec:online_algorithm} shows the incremental graph layout optimization strategy and the adaptive scheduling algorithm for dynamic input graphs.
Sec. \ref{sec:evaluation} evaluates the proposed approach and Sec. \ref{sec:conclusion} concludes.
The appendix provides the proofs of Theorem \ref{thm:correctness}, \ref{thm:approximation}, \ref{thm:convergence}, and \ref{thm:time_complexity}.

\section{Background and Related Work}
\label{sec:background}

This section briefly introduces GNN regarding its workflow, example models, and edge applications.
Related works are next reviewed.

\subsection{Graph Neural Network}
\label{sec:gnn}

Graph Neural Network (GNN) is a kind of DNN extrapolated on the graph domain, which combines graph embedding techniques and neural network operators to learn latent representations for graph-structural data \cite{zhou2018graph, wu2019comprehensive}.
The input to a GNN model is an attributed graph, where the vertices and links form a relational topology and each vertex attaches with a feature vector describing its physical properties.
For instance, a vertex in a social graph is a user, whose feature vector may depict the user's age, location, and preference, \textit{etc.}, and a link between users can be a fellowship.
The output of a GNN model is called the \textit{representation vector} or the \textit{embedding}, which are multi-dimension vectors that encode the latent information of vertices.
For the social graph just mentioned, a user vertex's embedding could be a vector of floating-point numbers that enciphers the user's original features.
Since the embedding has translated the rich semantics of the raw data into a well-formed numerical representation, it can be used for miscellaneous downstream analytics such as user identity categorization \cite{zhang2020deep}, friendship prediction \cite{sankar2021graph}, and community detection \cite{chen2018supervised}.

\begin{figure}[t]
\centerline{\includegraphics[width=0.9\linewidth]{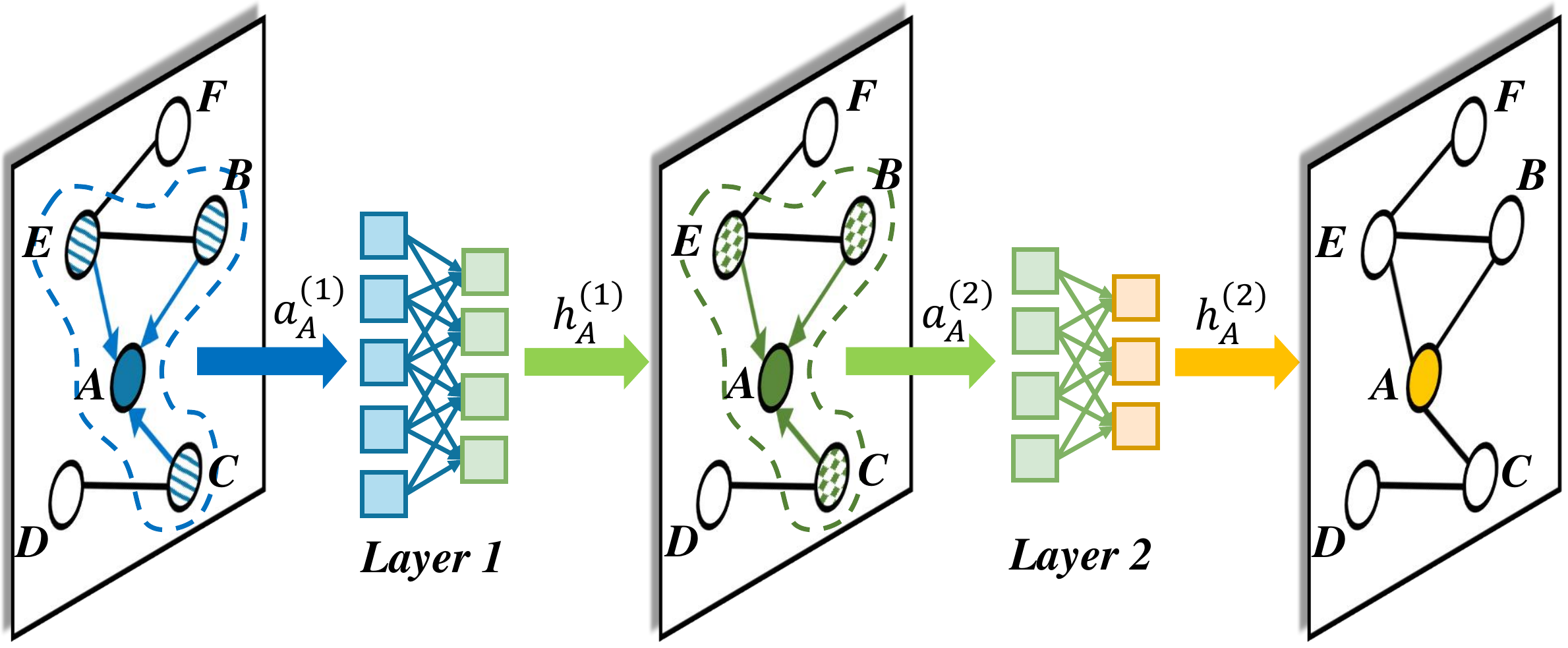}}
\caption{An instance of a two-layer Graph Neural Network architecture. To compute vertex \textit{A}'s embeddings, the model iteratively aggregates its neighbors, \textit{i.e.} \textit{B}, \textit{C}, and \textit{E}, and updates through neural activations.}
\label{fig:gnn}
\end{figure}

\textbf{Workflow.} 
Fig. \ref{fig:gnn} illustrates the inference workflow of a two-layer GNN architecture.
For the targeted vertex \textit{A}, the $k$-th GNN layer first \textit{aggregates} its neighbor vertices' feature vectors from \textit{B}, \textit{C}, and \textit{E}, and then \textit{updates} the aggregation $a^{(k)}_A$ through an activation to generate its representation vector $h^{(k)}_A$.
Every vertex in the graph follows the same procedure, \textit{i.e.} sharing the same aggregation and activation functions within the same layer, to generate the full graph embeddings.
By ingesting the previous layer's output into the successive layer as input, the layers are stacked to constitute the complete GNN model.

\textbf{Example models.}
We detail several typical GNN models to instantiate the aforementioned operations. 

\textit{GCN} \cite{kipf2016semi} is one of the first GNN models that draw convolutions on graph topology, and has been widely used in semi-supervised learning tasks.
Inference on its $k$-th layer is formalized as:
\begin{align}
    a^{(k)}_v = \sum_{u\in\mathcal{N}_v} h^{(k-1)}_u, \quad
    h^{(k)}_v = \sigma (W^{(k)} \cdot \frac{a^{(k)}_v+h^{(k-1)}_v}{|\mathcal{N}_v|+1}),
\end{align}
where $\mathcal{N}_v$ defines the set of $v$'s neighbor vertices, $W^{(k)}$ is the weights and $\sigma(\cdot)$ is an element-wise nonlinearity.

\textit{GAT} \cite{velivckovic2017graph} incorporates attention mechanism into graph convolutions to differentially weight the targeted vertex's neighbors. 
These vertex-varying weights are called attention parameters, denoted by $\eta^{(k)}_{vu}$, and are learned in model training.
For inference, the $k$-th layer consists of the following operations:
\begin{align}
    a^{(k)}_v = \sum_{u\in\mathcal{N}_v\cup \{v\}} \eta^{(k)}_{vu} W^{(k)} h^{(k-1)}_u, \quad 
    h^{(k)}_v = \sigma( a^{(k)}_v ). 
\end{align}

\textit{GraphSAGE} \cite{hamilton2017inductive} represents the inductive genre of GNN models.
It adopts a neighbor sampling technique to scale the model training on large graphs.
Its inference, however, fully collects the neighbor set and can be described as:
\begin{align}
    a^{(k)}_v = \frac{\sum_{u\in\mathcal{N}_v} h^{(k-1)}_u}{|\mathcal{N}_v|},
    h^{(k)}_v = \sigma( W^{(k)} \cdot (a^{(k)}_v, h^{(k-1)}_v) ). 
\end{align}
Note that GraphSAGE has variants with different aggregation functions. Here we select the mean aggregation version.

\textbf{Edge applications.}
Owing to their superior ability in abstracting graph data, GNNs have been employed in a wide range of real-world edge computing scenarios, especially for emerging IoT-driven applications.
For instance, traffic sensory networks can be naturally abstracted as a graph with roadside detectors as vertices and roads as links, where each vertex attaches a vector recording sensory data like traffic speed and occupancy.
Regarding such graphs, GNN models \cite{li2017diffusion, guo2019attention, wang2020traffic} are usually applied to perform traffic flow forecasting at a future time window.
Another example is location-based recommendation \cite{zhong2020hybrid, chang2020learning, yuan2020spatio} that employs GNN models to infer users' PoI, given their social connections and the graph of landmarks' spatial relations.
Since IoT devices are usually working nonstop, these GNN-based services are typically provisioned in a resident manner and process graph data streams continuously, which urgently emphasizes the necessity of a cost-effective solution for economic, efficient, and scalable GNN processing.

\subsection{Related Work}
\label{sec:related_work}

As an interdiscipline of GNN systems and edge intelligence, optimizing the cost of DGPE efficiently is still absent in the current context.
We now discuss the literature related to our problem in the following aspects.

\textbf{Distributed GNN systems.}
To support efficient GNN processing, a few frameworks have been developed in both research \cite{jia2020improving, tian2020pcgcn, thorpe2021dorylus, mohoney2021marius, cai2021dgcl, yan2020hygcn, zeng2022fograph} and industry communities \cite{hu2020featgraph, ma2019neugraph, zhu2019aligraph, zhang2020agl, gandhi2021p3, wang2021flexgraph} to optimize performance at different levels.
Towards general kernel-level optimizations, several works \cite{ma2019neugraph, wang2021flexgraph, wu2021seastar} endeavor to abstract a user-friendly programming model, and then apply a set of techniques under the unified interfaces by exploiting the execution primitives of GNN models.
For instance, FlexGraph \cite{wang2021flexgraph} identifies a two-dimensional categorization for various GNN models, and proposes strategies for hybrid workload balancing and overlapping processing. 
For scaling to many-GPU systems, a growing number of works \cite{zhang2020agl, gandhi2021p3, thorpe2021dorylus} has proposed to improve the parallelism across multiple servers, \textit{e.g.}, $P^3$ \cite{gandhi2021p3} designs a pipelined push-pull mechanism for accommodating more data at layered iteration, and Dorylus \cite{thorpe2021dorylus} introduces serverless functions to enable affordable yet massive parallelism with CPU servers.
For reducing communication during distributed runtime, many systems \cite{tripathy2020reducing, zheng2020distdgl, zhang2020agl} focus on improving graph layout to avoid bandwidth bottlenecks, while some \cite{cai2021dgcl} intend to design efficient routing algorithms for inter-server data exchange.
Although these systems achieve fair performance in hosting distributed GNN execution, they are all considered and designed for datacenter-level computing or with cloud assistance, supposing a powerful ensemble of unlimited computing capabilities and superior storage. 
These assumptions, however, fail to hold again for edge computing environments, where the computing resources are typically heterogeneous \cite{shi2016edge, zhou2019edge} and the proximity between clients and edge servers (\textit{w.r.t.} data collection) cannot be omitted.
To unleash the architectural benefits of edge computing, Fograph \cite{zeng2022fograph} first investigates GNN processing with vicinal fog servers, and proposes a distributed inference system for real-time serving. 
Yet it merely pursues performance indicators on latency and throughput, where a comprehensive cost model for DGPE still lacks, demanding a new formulation for analysis.

\textbf{Federated learning for GNN.}
Training GNN models under Federated Learning (FL) paradigm has recently attracted increasing attention in distributed edge deployment of GNNs \cite{he2021fedgraphnn, fu2022federated, liu2022federated}.
Unlike traditional training techniques that explicitly collect complete data and train at a dedicated location, FL allows distributed clients to maintain their sub-models (trained with local data) while learning a shared model collaboratively through gradient aggregation \cite{xia2021survey, diao2020heterofl, hong2021efficient}.
The design rationale behind FL generally puts privacy into the highest priority, targeting the scenario wherein the distributed graph data are possessed by different organizations \cite{wu2022federated, kairouz2021advances, yin2021comprehensive}.
Yet there are many works contemplate other dimensions for federated GNNs: accelerating the training convergence by tackling clients' heterogeneity and communication efficiency \cite{wang2022federatedscope, meng2021cross}, boosting the training accuracy upon the \textit{not Identically and Independently Distributed} (non-IID) graph data \cite{xie2021federated, wang2020graphfl, zheng2021asfgnn}, and ensuring the security and fairness of cross-silo data aggregation \cite{pei2021decentralized, jiang2022federated}, \textit{etc.}
Nonetheless, DGPE significantly diverges from the FL category in that:
1) DGPE considers distributed GNN inference processing services with a variety of cost optimization objectives, allowing attainable data sharing across edge servers, while FL focuses only on the GNN model training and particularly stresses user privacy by physical data isolation.
2) DGPE targets a technically distinct execution mechanism from FL, where DGPE only touches the distributed clients and edge servers, and processes GNN inference workload in parallel without the orchestration of centralized cloud servers (as required by FL).
Therefore, existing works on federated GNNs optimization cannot be straight exerted to the DGPE problem, necessitating a fire-new solution.

\textbf{Collaborative edge intelligence.}
To enable edge intelligence faster and greener, plenty of efforts have been made to parallelize DL models collaboratively.
A mainstream of existing literature \cite{kang2017neurosurgeon, teerapittayanon2017distributed, li2019edge, zeng2019boomerang} is to split the DL models into two parts, and collaborate their execution across clients and edge servers.
For instance, Edgent \cite{li2019edge} facilitates device-edge synergy with multi-exit models to achieve a tradeoff between inference accuracy and latency, and SPINN \cite{laskaridis2020spinn} further improves system robustness with various types of service level agreements considered.
Another line of works \cite{zhao2018deepthings, mao2017modnn, hadidi2018distributed, zeng2020coedge} exploits the data dimension by partitioning and dispatching inputs to multiple edge servers.
DeepThings \cite{zhao2018deepthings} proposes to partition the input image into grids, and applies an operator fusing technique to reduce communication across edge devices.
CoEdge \cite{zeng2020coedge} designs an adaptive workload partitioning strategy to jointly manage computation and communication throughout cooperation, achieving both energy efficiency and timely responsiveness. Nonetheless, these works all focus on traditional DL models like CNNs and RNNs, ignoring the emerging GNNs which have been widely used in edge scenarios.
Moreover, their formulation and optimization can not be directly applied to GNN given GNN's different computing pattern from other DNNs.

\textbf{Graph data placement.}
From the perspective of the input data structure, DGPE can be regarded as a modern mutant of distributed graph computing systems.
In this respect, there have been a lot of efforts towards optimizing system costs via improving the graph data placement \cite{jiao2014multi, jiao2014optimizing, yu2015location}.
For instance, in online social services, the server location assignment of graph data has been widely studied to meet various kinds of objectives like carbon footprint \cite{jiao2014multi} and Quality of Services (QoS) \cite{jiao2014optimizing}.
In collaborative applications deployment, researchers \cite{wang2020service} investigate the placement of functionally associated service entities among a board of servers, in order to minimize the total economic budget.
Nevertheless, the cost optimization for DGPE is yet under exploration, and existing arts consider neither GNN's distinctive computing characteristics nor the impact of potential input graph evolution, remaining a vacuum for GNN-oriented graph layout optimization.

\section{System Model}
\label{sec:system_model}

This section presents the system model for Distributed GNN Processing over heterogeneous Edge servers (DGPE) and the problem formulation of cost minimization.
Table \ref{tab:notations} summarizes the main notations used in our formulation.

\subsection{System Overview}
We target a DGPE service upon a multi-tier architecture as in Fig. \ref{fig:scenario}, where edge servers are geographically distributed in different areas and clients have their data individually collected from certain access points (\textit{e.g.}, 5G base station, IoT gateway) \cite{yang2019multi, yang20226g}.
For ease of formulation, we focus the GNN processing workload on inference.
Central to DGPE, we especially stress two types of graphs: 1) the \textit{edge network} that hosts distributed model execution, and 2) the \textit{data graph} formed by clients' associated data, which feeds the GNN model as the input graph.
We bridge them by further defining \textit{graph layout}.
To avoid ambiguity, we exclusively use $i$, $j$ to index the servers in the edge network, whereas $v$ and $u$ for the clients inside the data graph.

\textbf{Edge network.}
Consider an edge network $\mathcal{T}=(\mathcal{D}, \mathcal{W})$ that comprises a set $\mathcal{D}$ of edge servers, which are spatially distributed in a certain area (like a city) and are organized in a multi-tier architecture.
The edge servers in $\mathcal{D}$ are interconnected via a city-wide area network, and we use $\mathcal{W}=\{w_{ij}|i,j\in\mathcal{D}\}$ to characterize their connectivity where $w_{ij}=1$ if the edge servers $i$ and $j$ are mutually accessible and $w_{ij}=0$ or else.
Each edge server associates several APs (\textit{e.g.}, 5G base station, IoT gateway) that allow clients to participate in services.
Without loss of generality, the GNN processing workload on a single edge server is assumed to be containerized with lightweight virtualization techniques \cite{xiong2018extend, barbalace2020edge}, and can thus be flexibly allocated with dedicated computing resources and shared across edges.
This also enables system heterogeneity, \textit{i.e.}, the capability of GNN serving may vary across edge servers.

\textbf{Data graph.}
The input of the GNN execution is given by a data graph $\mathcal{G} = (\mathcal{V}, \mathcal{E})$, where the clients contribute as data points in $\mathcal{V}$ and $\mathcal{E} = \{e_{vu}|v,u \in \mathcal{V}\}$ indicates their topology.
Specifically, $e_{vu}=1$ if two clients $v$ and $u$ have links\footnote{To avoid ambiguity, we use \textit{links} to indicate relationships between vertices (clients) in the data graph, while exclusively leave \textit{edge} for edge servers.} in $\mathcal{G}$ and $e_{vu}=0$ otherwise.
Therefore, a client $v$'s direct neighbors is exactly $\mathcal{N}_v = \{u|e_{vu}=1\}$.
From the model aspect, a client $v$ initially issues a feature $h^{(0)}_v$ (which has a dimension size $s_0$), fed as the input vector for GNN processing as discussed in Sec. \ref{sec:gnn}.
As the execution proceeds, we label a vertex's embedding by $h^{(k)}_v$ in the $k$-th layer, with each in dimension size $s_k$.
Note that the clients can freely access or exit the services, and establish or abolish a connection with each other, which may incur topological changes on the data graph.
Albeit, due to the expressiveness of GNN models, the evolution is usually within a small extent \cite{skarding2021foundations}, where dramatic changes are not allowed.

\begin{table}[t] 
  \caption{Main notations used in our formulation.}
  \label{tab:notations}
  \centering
  \begin{tabular}{|c|l|}
    \hline 
    \textbf{Symbol} & \makecell[c]{\textbf{Description}} \\ \hline \hline
    $\mathcal{T}, \mathcal{D}, \mathcal{W}$ & \makecell[l]{The edge network $\mathcal{T}$ that consists of a set $\mathcal{D}$ of edge\\servers and a set $\mathcal{W}$ of their connections} \\ 
    $\mathcal{G}, \mathcal{V}, \mathcal{E}$ & \makecell[l]{The data graph $\mathcal{G} $ with vertices set $\mathcal{V}$ and link set $\mathcal{E}$}  \\ 
    $\pi, x_{vi}$ & \makecell[l]{The graph layout $\pi$ of binary variables $x_{vi}$ that identi-\\fies whether vertex $v$ is assigned to edge server $i$} \\
    $C$ & \makecell[l]{The total cost function $C$ that covers four types of cost\\factors in data collection $C_U$, GNN computation $C_P$,\\cross-edge traffic $C_T$, and edge server maintenance $C_M$ }\\ 
    $\mathcal{N}_v$ & The set of vertex $v$'s neighbors in $\mathcal{G}$\\
    $h^{(k)}_{v}$ & The feature vector of vertex $v$ at the $k$-th GNN layer\\
    $\mu_{vi}$ & Cost of uploading a client $v$' data to edge server $i$\\
    $\alpha_i, \beta_i, \gamma_i$ & Cost parameters in computing GNN at edge server $i$\\ 
    $\tau_{ij}$ & Cost of transferring a data unit across edges $i$ and $j$ \\ 
    $\rho_i, \varepsilon_i$ & Cost parameters of provisioning edge server $i$ \\ 
    $\mathcal{A}(v)$ & \makecell[l]{The auxiliary graph that comprises vertex $v$ and its\\associated neighbors} \\ 
    \hline
  \end{tabular}
\end{table}

\textbf{Graph layout.}
We use graph layout $\pi$ to couple the edge network $\mathcal{T}$ and the data graph $\mathcal{G}$ and define how distributed GNN processing is hosted over them.
Particularly, we apply binary variables $\pi=\{x_{vi}|v\in\mathcal{V}, i\in\mathcal{D}\}$ to characterize the data placement, such that $x_{vi}=1$ if client $v$'s data is processed at edge $i$ and $x_{vi}=0$ if not.
$x_{vi}$ complies with $\sum_{i \in \mathcal{D}}x_{vi}=1, \forall v \in \mathcal{V}$, as each client's data can appear in exactly one edge server\footnote{
While vertex-centric graph computing systems widely applied vertex replication across servers, the mainstream GNN systems yet maintain a link-centric execution paradigm without replicating vertex data \cite{scardapane2020distributed, jia2020improving, zhang2020agl}.
Our model thus follows this setting and leaves replication issues in future work.
}.
This can be relevant since the edge servers can readily share their collected graph data with each other during the runtime, and disabling replication can economize the storage overhead.
Remark that graph layout is a workload allocation descriptor beyond traditional data placement: in graph layout, each data unit, \textit{i.e.} vertex, contributes differentiated computation cost (depending on the number of its neighbors), while data placement usually treats data points as individual entities without concerning their associations.
Tuning graph layout can thus effectively optimize distributed GNN processing: by navigating data streams of clients to proper edge servers, we can jointly minimize system cost in terms of data collection, GNN computation, cross-edge traffic, and edge server maintenance.
We then elaborate on these cost factors in the next subsection.

\subsection{Cost Factors}
We now formulate four types of system costs following the lifecycle of GNN processing in DGPE deployment.

\textbf{Data collection.}
The flow of distributed GNN processing begins with collecting graph data from distributed clients.
From a client $v$'s perspective, its data $h_v$ is uploaded to the edge server $i$ via the nearest AP, denoted by $\hat{i}$, in its physical vicinity, as illustrated in Fig. \ref{fig:uploading}.
Assuming the uploading cost of a single vertex' data is quantified by $\mu_{vi}$ from client $v$ to edge $i$, the combined cost for an edge server's data collection is $\sum_{v\in\mathcal{V}} \mu_{vi} x_{vi}$.
The total cost for data collection is therefore given by
\begin{align}
    C_U = \sum_{i\in\mathcal{D}} \sum_{v\in\mathcal{V}} \mu_{vi} x_{vi}. \label{eq:uploading}
\end{align}
Note that the communication between a client and its associated AP is an inevitable procedure for every client within the scope, and is hence irrelevant to the placement decision in essence.
We thus omit such cost for simplicity.

\textbf{GNN computation.}
According to the execution pattern introduced in Sec. \ref{sec:gnn}, to compute inference through a GNN model consists of aggregation and update steps.
For the former, the aggregation of embeddings $\sum_{u\in\mathcal{N}_v}h^{(k-1)}_v$ is linear to the number of a vertex's neighbors $|\mathcal{N}_v|$ and the size $s_{k-1}$ of input feature vectors (of the previous layer $k-1$).
Given a unit cost $\alpha$ of summarizing two vectors, the computation cost of aggregation can be calculated in $\alpha |\mathcal{N}_v| s_{k-1}$.
For the latter, the update step first weights aggregated results with $W^{(k)}$ and next applies an activation function $\sigma(\cdot)$.
With regard to its computing properties \cite{zhao2021cm}, the costs of these two arithmetical operations are linear to the multiplication of the sizes of input/output feature vectors, \textit{i.e.}, $s_{k-1}s_{k}$, and the obtained embedding vector size $s_{k}$, respectively.
Given $\beta$ as the unit cost for matrix-vector multiplication and $\gamma$ for applying activation, the cost of update is $\beta s_{k-1} s_{k} + \gamma s_{k}$.
Putting the two steps together in a $K$-layer GNN model, the computation cost of computing vertex $v$ on edge server $i$ is
\begin{align}
    C_P(v,i) = \sum_{k\in\mathcal{K}} ( \alpha_i |\mathcal{N}_v| s_{k-1} + \beta_i s_{k-1} s_{k} + \gamma_i s_{k} ),\label{eq:computation_v}
\end{align}
where $\alpha_i$, $\beta_i$, and $\gamma_i$ can be heterogeneous across edge servers.
Also, these parameters may vary for different GNN models, as they may apply different aggregation functions and neural network operators.
To sum up the total computation cost for DGPE, we apply the above formula (\ref{eq:computation_v}) to all vertices on all edge servers, and obtain 
\begin{align}
    C_P = \sum_{i\in\mathcal{D}} \sum_{v\in\mathcal{V}} C_P(v,i) x_{vi}. \label{eq:computation}
\end{align}

\begin{figure}[t]
  \centering
  \includegraphics[width=0.8\linewidth]{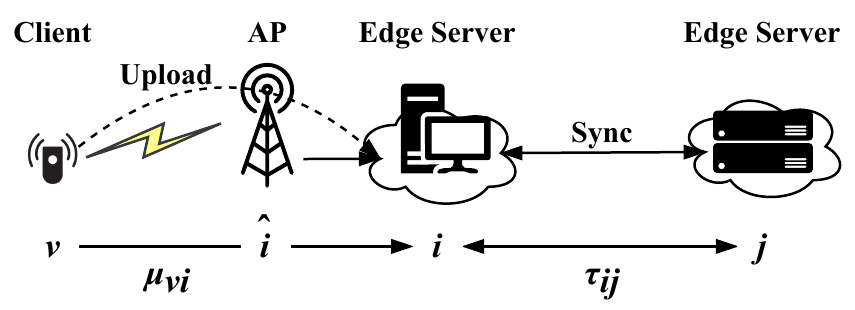}
  \caption{Illustration of the data collection cost and cross-edge traffic cost for a single client's data.
  }
  \label{fig:uploading}
\end{figure}

\textbf{Cross-edge traffic.}
To accurately render GNN processing during the distributed runtime, besides computing GNN over their own resident graphs, edge servers need to pull vertices' data from each other.
This attributes to the unique computing pattern of GNN, where each vertex's computation necessarily relies on its neighbors' data.
For the data graph spans over multiple edge servers, the data to be transferred is the feature vectors from vertices that have connections across edge servers.
For instance, in Fig. \ref{fig:gnn}, \textit{B}'s feature vectors is ought to be sent to the edge server that \textit{A} locates if \textit{A} and \textit{B} are placed at different servers, when computing vertex \textit{A}'s embedding.
In our modeling, these cross-edge connections can be identified by $e_{uv} w_{ij} x_{vi} x_{uj}$, $\forall u,v \in \mathcal{V}, \forall i,j \in \mathcal{D}$, where $e_{uv} \in \mathcal{E}$ and $w_{ij} \in \mathcal{W}$ check whether there is subsistent links of $\langle u,v \rangle$ or $\langle i,j \rangle$, respectively.
Existing distributed GNN systems \cite{zeng2022fograph, zhu2019aligraph, zhang2020agl} typically leverage the Bulk Synchronous Parallel (BSP) model \cite{valiant1990bridging} to orchestrate parallelism over distributed servers, where a synchronization step is triggered when inter-server data exchange is needed.
We adopt this setting and refer to such synchronization as cross-edge traffic, with a unit communication cost $\tau_{ij}$ for edge pairs $\langle i,j \rangle$ as depicted in Fig. \ref{fig:uploading}.
A large $\tau_{ij}$ occurs when the edge servers $\langle i,j \rangle$ have a long routing distance in the network hierarchy, and it will be set as infinity if they fail to establish a valid connection.
With this parameter, the total cross-edge traffic cost is calculated by
\begin{align}
    C_T = \sum_{i\in\mathcal{D}} \sum_{j\in\mathcal{D}} \sum_{v\in\mathcal{V}} \sum_{u\in\mathcal{V}} \tau_{ij} e_{uv} w_{ij} x_{vi} x_{uj}. \label{eq:traffic}
\end{align}

\textbf{Edge server maintenance.}
Provisioning a DGPE service instance at an edge server requires loading an image containing the GNN model, booting it, and sustaining its runtime, and thus the maintenance cost is to be paid.
In principle, we divide it into two types.
One is data-dependent, such as the monetary cost for hosting the computing and storage resources, which is related to the number of vertices at the server.
The other is data-independent which is typically a one-shot expenditure like the cost of container launching and machine cooling.
Let $\rho_i$ be the data-dependent cost per vertex and $\varepsilon_i$ for the data-independent one at edge $i$, we then have the edge server maintenance cost as
\begin{align}
    C_M = \sum_{i\in\mathcal{D}} (\sum_{v\in\mathcal{V}} \rho_{i} x_{vi} + \varepsilon_i). \label{eq:maintanence}
\end{align}

\subsection{Cost-Efficient Graph Layout Optimization Problem}

Summarizing the above cost factors yields the total cost $C$, defined in Eq. (\ref{eq:total_cost}), acting as the objective function.
\begin{align}
    C = C_U + C_P + C_T + C_M. \label{eq:total_cost}
\end{align}

Formally, given a GNN service over the edge network $\mathcal{T}$ with the input data graph $\mathcal{G}$, the cost optimization problem $\mathcal{P}$ for DGPE is to decide a graph layout $\pi$ such that the total system cost $C$ is minimized:
\begin{align}
    \mathcal{P}: \quad \min_{\pi} \ &\mathcal{C}(\pi|\mathcal{T}, \mathcal{G}), \label{eq:objective} \\
    \text{s.t.} \ & \sum_{i \in \mathcal{D}} x_{vi} = 1, \forall v \in \mathcal{V}, \tag{\ref{eq:objective}{a}} \label{eq:constraint_a} \\
    & x_{vi} \in \{0, 1\}, \forall v \in \mathcal{V}, \forall i \in \mathcal{D}, \tag{\ref{eq:objective}{b}} \label{eq:constraint_b}\\
    & \pi = \{x_{vi}| v \in \mathcal{V}, i \in \mathcal{D}\}. \tag{\ref{eq:objective}{c}} \label{eq:constraint_c}
\end{align}

The constraints (\ref{eq:constraint_a})-(\ref{eq:constraint_c}) state that the decision variables $x_{vi}$ in $\pi$ should place each client's data to exactly one edge server.  
Our formulation in $\mathcal{P}$ provides a comprehensive model that is flexible and compatible to capture various types of objectives in DGPE.
For instance, by interpreting parameters like $\alpha_i$ and $\beta_i$ as power consumption per operation, the model can expressively characterize how much energy is used for the service runtime \cite{jiao2014multi, zhou2019online}.
Alternatively, specifying unit costs in the monetary base navigates the model toward an affordable economic consideration \cite{jiao2014optimizing}.
After aligning a unified indicator for all parameters, one can further assign proper weights to cost factors to steer the performance optimization on specific aspects of the system.

Despite its generalizability, optimizing $\mathcal{P}$ is non-trivial due to its complex combinatorial nature and the decision variables' nonlinear, intertwined relationships implicated in GNN's neighbor-aggregation mechanism.
The heterogeneity of edge servers, along with their diverse communication overhead, further complicates the cost-optimized graph layout scheduling.
On the one hand, a server with higher computing capability typically affords a heavier workload with lower cost, and hence could be assigned more vertices for lower system cost.
On the other hand, if the bandwidth of this server is severely limited, the considerable traffic cost will hinder the assignment of linked vertices to it, which conversely constrains the capability-first scheduling.
In essence, Theorem \ref{thm:np-hard} shows that the corresponding problem is NP-hard, indicating an absence of polynomial algorithms unless P=NP.

\begin{theorem}[Hardness] \label{thm:np-hard}
The cost-efficient graph layout optimization problem $\mathcal{P}$ is NP-hard.
\end{theorem}
\begin{proof}
We prove the NP-hardness of $\mathcal{P}$ by reducing from the Minimum Weighted Set Cover (\textit{MWSC}) problem, which is known to be NP-hard \cite{chvatal1979greedy}.
Formally, given a universal set $\mathcal{U}$ and a family $\mathcal{S}$ of subsets of $\mathcal{U}$, the \textit{MWSC} problem is to find a set $\mathcal{I} \subseteq \mathcal{S}$ such that all elements in $\mathcal{U}$ is covered by $\mathcal{I}$ while the summed weights of the elements in $\mathcal{I}$ is minimized.
Consider an instance of $\mathcal{P}$, denoted by $\mathcal{P}_0$, with $\tau_{ij}=0$ for every edge pair $\langle i,j \rangle$ (\textit{i.e.} the cost of cross-edge traffic is omitted), we can build a polynomial transformation by recognizing $\langle \mathcal{V}, \mathcal{D} \rangle$ as the universe $\mathcal{U}$, and setting $\mathcal{S} = \{\langle v, i \rangle| v \in \mathcal{V}, i\in \mathcal{D}\}$ and $\mathcal{I} = \mathcal{S}$.
Here the weight of element $\langle v, i \rangle$ is given by $C_U(v,i) + C_P(v,i) + C_M(v,i)$, \textit{i.e.}, the total cost of data collection, GNN computation, and (data-dependent) server maintenance when placing client $v$'s data to edge server $i$.
Since the constraints in $\mathcal{P}_0$ essentially align with the requirement of set covering, $\mathcal{P}_0$ is exactly a \textit{MWSC} problem.
Given that \textit{MWSC} is NP-hard, the above polynomial reduction indicates that $\mathcal{P}$ is also NP-hard, which completes the proof.
\end{proof}

\section{Approximate Algorithm Design for Static Input Graphs}
\label{sec:offline_algorithm}

Provided with $\mathcal{P}$'s NP-hardness, we intend to design an efficient approximate algorithm by exploring and exploiting the structural properties of $\mathcal{P}$.
Concretely, in this section, we consider the case of  $\mathcal{P}$ with static data graphs, where the input graph's topology is invariant (while the clients' input feature data can be varying).
This is relevant to the scenarios wherein the client devices are fixed (\textit{e.g.}, traffic sensory network in which its sensors are settled at dedicated locations) or the clients' relationships can be statistically modeled as a static graph. 
In what follows, we first introduce our insights on the objective function's properties, and next present our algorithm with theoretical performance analysis.

\subsection{Motivation}
On the objective of $\mathcal{P}$, \textit{i.e.} the total cost function $C$, we observe the properties stated in Theorem \ref{thm:pseudo-boolean} and Theorem \ref{thm:submodular}.

\begin{theorem}[Pseudo-boolean quadratic] \label{thm:pseudo-boolean}
The cost function $C$ is a quadratic pseudo-boolean function.
\end{theorem}
\begin{proof}
We unfold the cost function $C$ through identifying three components $C_0$, $C_1$ and $C_2$:
\begin{align}
    C = & C_U + C_P + C_T + C_M \notag \\
    = & \overbrace{\sum_{i\in\mathcal{D}} \varepsilon_i }^{C_0}
    + \overbrace{\sum_{i\in\mathcal{D}} \sum_{v\in\mathcal{V}} [\mu_{vi} + C_P(v,i) + \rho_{i} ] x_{vi}}^{C_1} \notag \\
    & + \overbrace{\sum_{i\in\mathcal{D}} \sum_{j\in\mathcal{D}} \sum_{v\in\mathcal{V}} \sum_{u\in\mathcal{V}} \tau_{ij} e_{uv} w_{ij} x_{vi} x_{uj}}^{C_2}, \label{eq:cost_three}
\end{align}
where $C_0$, $C_1$, and $C_2$ are the constant, linear, and quadratic terms toward the variables $x_{vi}$ in $\pi$, respectively.
Since every element in $\pi$ is binary and the value of $C$ lies in $\mathbb{R}$, the cost function essentially constructs a mapping of $\{0,1\}^{|m\times n|} \rightarrow \mathbb{R}$ and hence it is a quadratic pseudo-boolean function.
\end{proof}

\begin{theorem}[Submodularity] \label{thm:submodular}
The cost function $C$ is a submodular function.
\end{theorem}
\begin{proof}
The cost function $C$ is submodular if and only if for every subset $\mathcal{X}, \mathcal{Y} \subseteq \mathcal{V}$ with $\mathcal{X} \subseteq \mathcal{Y}$ and every vertex $v \in \mathcal{V} \setminus \mathcal{Y}$, we have $F(\mathcal{X}, v) \geq F(\mathcal{Y}, v)$, where $F(\mathcal{X}, v)$ and $F(\mathcal{Y},v)$ are the marginal cost function for $\mathcal{X}$ and $\mathcal{Y}$, respectively \cite{fujishige2005submodular}:
\begin{align}
    F(\mathcal{X}, v) &= C(\pi'_{\mathcal{X}}|\mathcal{T}, \mathcal{X} \cup \{v\}) - C(\pi_{\mathcal{X}}|\mathcal{T}, \mathcal{X}),\\
    F(\mathcal{Y}, v) &=  C(\pi'_{\mathcal{Y}}|\mathcal{T}, \mathcal{Y} \cup \{v\}) -  C(\pi_{\mathcal{Y}}|\mathcal{T}, \mathcal{Y}).
\end{align}

Here we use $\pi_{\mathcal{X}}$, $\pi_{\mathcal{Y}}$ and $\pi'_{\mathcal{X}}$, $\pi'_{\mathcal{Y}}$ to indicate the graph layout of $\mathcal{X}$, $\mathcal{Y}$ before and after adding vertex $v$, respectively.
We enforce $\pi_{\mathcal{X}} \subseteq \pi'_{\mathcal{X}}$ and $\pi_{\mathcal{Y}} \subseteq \pi'_{\mathcal{Y}}$ to ensure the placement of origin vertices is invariant, and impose $\pi'_{\mathcal{X}} \subseteq \pi'_{\mathcal{Y}}$ so that the newly-added vertex $v$, to either $\mathcal{X}$ or $\mathcal{Y}$, is assigned to the same edge server.

We now show the four factors in $C$ all satisfy the above sufficient and necessary condition, case by case.

\textit{1) Data collection cost} $C_U$.
As the data collection cost is counted individually for each vertex, the marginal cost for appending a new vertex is obviously the same for both $\mathcal{X}$ and $\mathcal{Y}$.
Hence, $F_U(\mathcal{X}, v) = F_U(\mathcal{Y},v)$ and $C_U$ is submodular.

\textit{2) GNN computation cost} $C_P$.
As discussed in Sec. \ref{sec:gnn}, to compute vertex $v$ through a GNN layer requires a local auxiliary graph that composes the vertex itself and its neighbors.
We denote such an auxiliary graph for a single GNN layer as $\mathcal{A}(v) = \mathcal{N}_v \cup v$, and consequently we have $\mathcal{A}({\mathcal{X}}) = (\bigcup_{v \in \mathcal{X}}\mathcal{N}_v) \cup \mathcal{X}$ and $\mathcal{A}({\mathcal{Y}}) = (\bigcup_{v \in \mathcal{Y}}\mathcal{N}_v) \cup \mathcal{Y}$.
Since $\mathcal{X} \subseteq \mathcal{Y}$, for each layer $\mathcal{A}({\mathcal{X}}) \subseteq \mathcal{A}({\mathcal{Y}})$ holds.
Therefore, for the new vertex $v$ we have $\mathcal{N}_v \setminus \mathcal{A}({\mathcal{X}}) \supseteq \mathcal{N}_v \setminus \mathcal{A}({\mathcal{Y}})$.
For any additional vertex $v \in \mathcal{V} \setminus \mathcal{Y}$, $\mathcal{A}({\mathcal{X} \cup \{v\}}) = \mathcal{A}({\mathcal{X}}) \cup \mathcal{N}_v$.
Then for the marginal cost function for $C_P$ on a dedicated edge and GNN layer, we have
\begin{align}
    F_P(\mathcal{X}, v) 
    &= C_P(\pi'_{\mathcal{X}}|\mathcal{T}, \mathcal{X} \cup \{v\}) - C_P(\pi_{\mathcal{X}}|\mathcal{T}) \notag\\
    &= \sum_{u\in\mathcal{A}({\mathcal{X} \cup \{v\}})}C_P(u) - \sum_{u\in\mathcal{A}({\mathcal{X}})}C_P(u) \notag\\
    &= \sum_{u\in\mathcal{N}_v \setminus \mathcal{A}({\mathcal{X}})}C_P(u) \notag\\
    &\geq \sum_{u\in\mathcal{N}_v \setminus \mathcal{A}({\mathcal{Y}})}C_P(u) \notag\\
    &= F_P(\mathcal{Y},v).
\end{align}
This thereby results in that $C_P$ is submodular.

\textit{3) Cross-edge traffic cost} $C_T$.
The cross-edge traffic cost is induced when two vertices that are separately located at different edges have a relationship.
As we insert a vertex $v$ to $\mathcal{X}$, there are two circumstances between $v$ and the vertices in $\mathcal{X}$, \textit{i.e.} connected or unconnected. 
For the former, notice that a budget of traffic cost will be paid if the newly-added connection(s) bridge two edge servers, or else the connection(s) are within the same edge server and the total traffic cost stays changeless.
Both cases will incur the same marginal traffic cost for $\mathcal{X}$ and $\mathcal{Y}$ since $\mathcal{X} \subseteq \mathcal{Y}$.
For the latter, the isolated vertex will not add additional traffic cost.
In summary, $F_T(\mathcal{X}, v) = F_T(\mathcal{Y},v)$ and $C_T$ is submodular.

\textit{4) Edge server maintenance cost} $C_M$.
Analogous to the data collection cost, attaching a new vertex $v$ introduces a unit of supplementary cost for both $\mathcal{X}$ and $\mathcal{Y}$, since $v$ is placed at the same edge server provided with $\pi'_{\mathcal{X}} \subseteq \pi'_{\mathcal{Y}}$.
Therefore, $F_M(\mathcal{X}, v) = F_M(\mathcal{Y},v)$ and $C_U$ is submodular.

The discussion above shows all cost factors in $C$ are submodular, which directly follows that $C$ is a submodular function as $C$ is a linear summation of them \cite{krause2014submodular}.
\end{proof}

The above theorems reveal the cost function $C$'s structural properties of pseudo-boolean quadratic and submodularity, inherently implicated in GNN's computing pattern of neighbor aggregation.
This inspires two insights on solving $\mathcal{P}$.
First, the coexistence of these two properties indicates the cost function is representable by a graph when there are exactly two edge servers, with vertices and links in it attaching variables' values, according to Kolmogorov's theorem \cite{kolmogorov2004energy}.
Upon such an auxiliary representation, we can therefore convert the cost optimization problem (for two edge servers) into a minimum graph cut issue, utilizing existing efficient max-flow algorithms for efficient solving.
More concretely, let the two edge servers respectively be the source \textit{s} and terminal $t$, to schedule a cost-minimized graph layout among the edges and a group of clients is equivalent to finding the minimum \textit{s-t} cut in the auxiliary graph covering these clients and edge servers.
Theorem \ref{thm:correctness} provides a formal proof and we will explain more details in the description of our algorithm (Sec. \ref{sec:offline_algorithm_detail}).

Second, the implicit submodularity of the problem induces a greedy heuristic, which has been proven to be a satisfactory means with a bounded approximation ratio for submodular optimization \cite{krause2014submodular}.
Combining the two insights inspires us to transform the cost minimization problem for multiple edge servers into a series of minimum cut issues on individual pairs of edge servers, and greedily search for a feasible graph layout to approach the optimal solution of $\mathcal{P}$.

\subsection{Cost Optimization via Iterative Graph Cuts}
\label{sec:offline_algorithm_detail}

We now explain our graph-cut based algorithm, called GLAD-S, for solving $\mathcal{P}$ with \textbf{S}tatic input graphs, as described in Algorithm \ref{algo:offline}.
Fig. \ref{fig:graph-cut} describes an instance of GLAD-S's procedure.
The key idea of GLAD-S is to iteratively traverse the pairs of every two edge servers and find the minimum \textit{s-t} cut for each pair, where each obtained minimum cut set reflects an optimal assignment of involved clients to the pair of the selected two edge servers.
Its performance analysis is provided in the next subsection.

Algorithm \ref{algo:offline} accepts the edge network $\mathcal{T}$, the data graph $\mathcal{D}$, the parameterized cost function $C$, and an operator-defined measurement of convergence $R$ as input, and aims at exporting an optimized graph layout $\pi$ that minimizes $C(\pi)$.
The algorithm begins with randomizing a graph layout $\pi$ by assigning vertices to arbitrary edge servers, acting as the starting point of searching (Fig. \ref{fig:graph-cut}(b)).
A variable $r$ is defined to count how many checking times have been taken to judge a valid convergence, and is initialized using 0 (Line 2).

We now dive into iterative graph cuts as long as $r \leq R$.
For each iteration, we first select a pair of connected edge servers $\langle i , j \rangle$ with the minimum visited times and log the selection (Line 4).
If there are multiple pairs sharing the same minimum visited (selected) times, we select an arbitrary one.
The rationale behind using the visited times as guidance is to avoid redundant attempts on yet-traversed pairs, so that the exploration space is expanded as large as possible for potential exploitation.
With respect to the selected pair of edge servers $\langle i, j\rangle$, an auxiliary graph $\mathcal{A}(i,j)$ is constructed in two steps:
1) generate a graph by connecting both edge servers $i$ and $j$ to every vertex (client) that has been associated with one of them, as visualized in Fig. \ref{fig:graph-cut}(c), and 
2) assign an appropriate weight to each connection in this auxiliary graph.
The weight between a vertex $v$ and an edge server $i$ is calculated by summarizing the unary cost $C_1$ in Eq. (\ref{eq:cost_three}) and the side-effect cost in terms of $v$ and $i$, which reflects how much cost that client $v$ contributes to the system if it is placed at edge $i$.
Here the side-effect cost means the external traffic cost caused by the links that potentially astride $i$ and other edge servers (except $j$).
For instance, the weight of $\langle v, i \rangle$ in Fig. \ref{fig:graph-cut}(c) should include two parts: 1) the cost of data collection, GNN computation, and server maintenance when client $v$ is located at edge $i$, and 2) the traffic cost between edge $i$ and $k$ (since a part of $v$'s neighbors are placed at edge server $k$).
The weight between two vertices $v$ and $u$ is given by $C_T(i,j)$, indicating how much traffic cost will be paid if $v$ and $u$ are respectively placed at edge servers $i$ and $j$.
Given that $C_0$ in Eq. (\ref{eq:cost_three}) is a constant, the above weight instantiation covers all cost factors related to our decision variables, and thus can completely express the objective of $\mathcal{P}$.

\begin{algorithm}[t] 
\caption{GLAD-S: GLAD for Static input data graph} 
\label{algo:offline} 
\begin{algorithmic}[1] 
    \REQUIRE ~~\\ 
    $\mathcal{T}$: The edge network $(\mathcal{D}, \mathcal{W})$\\
    $\mathcal{G}$: The data graph $(\mathcal{V}, \mathcal{E})$\\
    $C(\pi)$: The system cost function for $\pi$\\
    $R$: The measurement of convergence\\
    \ENSURE ~~\\ 
    $\pi$: Optimized graph layout $\{x_{vi}| v \in \mathcal{V}, i \in \mathcal{D}\}$
    \STATE Initialize a randomized graph layout $\pi$
    \STATE $r \leftarrow 0$
    \WHILE{$r \leq R$}
        \STATE Select a pair of connected edge servers $\langle i,j \rangle$ with the minimum visited times
        \STATE Construct an auxiliary graph $\mathcal{A}(i,j)$ \textit{w.r.t.} $\langle i,j \rangle$ \\
        \STATE Solve the minimum s-t cut of $\mathcal{A}(i,j)$
        \STATE Build a graph layout $\pi'$ according to Eq. (\ref{eq:mapping}) from the obtained minimum cut set
        \IF{$C(\pi') < C(\pi)$}
            \STATE $\pi \leftarrow \pi'$
            \STATE $r \leftarrow 0$
        \ELSE
            \STATE $r \leftarrow r + 1$
        \ENDIF
    \ENDWHILE
    \RETURN $\pi$
\end{algorithmic}
\end{algorithm}

\begin{figure*}[t]
  \centering
  \includegraphics[width=\linewidth]{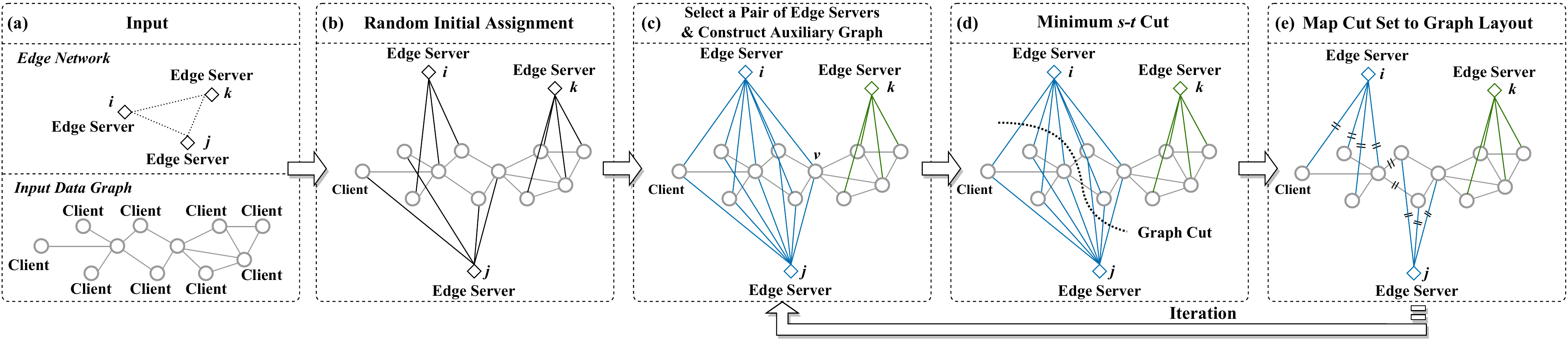}
  \caption{Instance of graph-cut based graph layout optimization. Given an edge network and an input data graph (a), GLAD-S first initializes a graph layout via random assignments (b). Next, it selects an arbitrary pair of edge servers and constructs an augment graph by connecting both edges to every vertex that has been associated with one of them (c). Minimum \textit{s-t} cut is then performed on the augment graph and obtained a set of graph cuts (d). Based on this cut set, it builds a graph layout for the edge servers and the client inside the augment graph, where marked links indicate assignments (e). GLAD-S will pass the obtained result to step (b) as the next iteration for finding a better graph layout until convergence.
  }
  \label{fig:graph-cut}
\end{figure*}

Upon the auxiliary graph $\mathcal{A}(i,j)$, we can always find the set of minimum cuts by identifying $i$ as the source and $j$ as the sink and running the max-flow min-cut algorithm.
Since the minimum \textit{s-t} cut ensures the graph to be split exactly into two parts, the constraints in $\mathcal{P}$ are also guaranteed.
Fig. \ref{fig:graph-cut}(d) marks the selected links by a dashed curve, where other links in the auxiliary graph are not in the cut set.
With the obtained cuts, we apply the mapping policy in Eq. (\ref{eq:mapping}) to transform them into corresponding graph layout $\pi'$ (Fig. \ref{fig:graph-cut}(e)), and examine the attained cost $C(\pi')$.
\begin{align}
\label{eq:mapping}
x_{vi} = \left\{ 
    \begin{array}{ll}
        1, & \textrm{if the link $\langle v,i \rangle$ is in the cut set,}\\
        0, & \textrm{otherwise.}
    \end{array} \right.
\end{align}
If $C(\pi') < C(\pi)$, then this round's attempt finds an improved graph layout and subsequently updates.
The available convergence-checking times $r$ is also reset to $0$ for prospective exploring rounds.
Otherwise, the attempt of this round fails to update $\pi$ and spends one checking opportunity, approaching the convergence one step closer (Line 12).
The algorithm then goes into another iteration by passing the existing graph layout to another graph-cut attempt, \textit{i.e.} from Fig. \ref{fig:graph-cut}(e) to Fig. \ref{fig:graph-cut}(c).
The overall loop terminates until exhausting all convergence checking times $R$, meaning no further cost reduction can be made under the attempt tolerance.
The obtained graph layout $\pi$ is consequently returned.

\textbf{Discussion.}
Algorithm \ref{algo:offline}'s convergence speed and obtained results are largely impacted by its initialization, traversing strategy, and the configuration of $R$.
For initialization, while our pseudocode uses randomization for generality, one can enhance the algorithm by manually devising a generative scheme with prior knowledge for the initial graph layout.
As an example, if the data collection cost of vertices data is exceedingly exorbitant and dominates the total system cost, initializing a graph layout via an uploading-first tactic (which greedily places vertices to edge servers that minimizes $C_U$) would expedite the convergence.
For the traversing strategy, instead of stochastic searching, we guide the selection of prospective edge pairs using the visited times as the proxy, in order to explore as many edge servers as possible for more searching possibilities.
If traversing randomly, however, the algorithm may encounter duplicated edge pairs that have been repeatedly checked, squandering the attempting opportunities (given by $R$) while increasing the probability of falling into a local optimum.
The configuration of $R$ essentially tunes the tradeoff between scheduling overhead and results' quality.
A larger $R$ claims stricter convergence measurement, yielding a better result while requiring more iteration times.
To prioritize the preference on the total system cost, we intend to find the optimal graph layout reachable by GLAD-S.
Since for each iteration, the space of available edge pairs to be selected is $\frac{|\mathcal{D}|(|\mathcal{D}|-1)}{2}$ (where $|\mathcal{D}|$ is the number of edge servers), we set $R=\frac{|\mathcal{D}|(|\mathcal{D}|-1)}{2}$ to ensure an exhaustive searching for the optimum.
Despite the default setting, we will also explore the impact of $R$ on system performance in our evaluation (Sec. \ref{sec:sensitivity}).

\subsection{Performance Analysis}

We present a rigorous analysis of GLAD-S's performance, showing that it has parameterized constant approximation ratio, as well as proving its convergence and time complexity in the following theorems.
Their proofs are provided in the appendix.

\begin{theorem}[Equivalence] \label{thm:correctness}
Given an auxiliary graph $\mathcal{A}(i,j)$ for optimizing the assignments between the
edge servers $i$ and $j$ and their associated clients, solving the minimum \textit{s-t} cut on $\mathcal{A}(i,j)$ can find the cost-minimized graph layout for the edge servers and clients in $\mathcal{A}(i,j)$.
\end{theorem}

\begin{theorem}[Approximation ratio] \label{thm:approximation}
Let $\pi$ and $\pi^*$ be the graph layout produced by GLAD-S and the global optimal placement, respectively. The upper bound of $C(\pi)$ is given by $2\lambda C(\pi^*) + \epsilon$, where $\lambda = \frac{\max_{i,j\in\mathcal{D}}\tau_{ij}}{\min_{i,j\in\mathcal{D}}\tau_{ij}}$ and $\epsilon = \sum_{i\in\mathcal{D}} \varepsilon_i$.
\end{theorem}

\begin{theorem}[Convergence] \label{thm:convergence}
GLAD-S is guaranteed to converge.
\end{theorem}

\begin{theorem}[Time complexity] \label{thm:time_complexity}
The time complexity of GLAD-S is $O((V+2)(E+2V)(V+E)R)$, where $V = |\mathcal{V}|$ and $E = |\mathcal{E}|$.
\end{theorem}

\section{Adaptive Algorithm Design for Evolved Input Graphs}
\label{sec:online_algorithm}

The discussion above focuses on $\mathcal{P}$ with static input data graphs.
However, there are still many edge deployments undertaking GNN services with input graphs of dynamic topology.
For instance, real-world social networks are usually evolving over time, with incoming users and changing friendships.
In this section, we concentrate on design algorithms for $\mathcal{P}$ with such evolved input graphs.
We will first present formal modeling of graph evolution, and next explain our incremental graph layout optimization as well as the adaptive scheduling algorithm.

\subsection{Modeling Graph Evolution}

Without loss of generality, let the DGPE system works in a time-fragmented fashion within a large time span of $t \in \{1,2,\cdots\}$, we denote the data graph at time $t$ as $\mathcal{G}(t)$.
In light of the traits of GNN services, a minimum time slot could be in the length of minutes, since practical deployment usually profiles system information periodically \cite{zheng2020distdgl, zhang2020agl} and input graph changes are thereafter detected.
Within each time slot, the system, including both the data graph and the edge network, is considered invariant, meaning that all cost-related parameters are fixed.

The evolution of an input data graph can be divided into four categories, in terms of inserting/deleting a vertex/link, as exemplified in Fig. \ref{fig:graph_dynamics}.
Their physical meaning can obviously be interpreted in edge scenarios:
for vertices' changes, an insertion means a new participant in GNN services, whereas a deletion can be a departure;
for links' changes, appending a new connection indicates a fresh relationship, and losing a bond can happen when a user leaves.
Due to GNN's expressiveness, all these modifications are restricted to be slight and within a number of candidates \cite{skarding2021foundations}, since a drastic change may significantly decline the inference accuracy and thus disable the service.

When the above changes take place across two time slots, the system performance may drop dramatically, even with an optimal graph layout in static optimization.
For instance, introducing a new client to the service requires arranging it to an edge server, and appending a new link between present vertices located at different edge servers incurs additional traffic cost.
All of these variations demand a graph layout update.
While the total cost function can be optimized separately in every time slot by invoking GLAD-S algorithm repeatedly, their graph layout outcome is very possibly inconsistent, leading to the issue of vertices migration.
However, as we move vertices across edge servers, their corresponding clients undergo interruptions, waiting for data replication, states transfer, and service rebooting, which may severely affect user experiences and QoS. 
This motivates us to contrive a lightweight strategy to minimize impacts on participating clients.

\begin{figure}[t]
  \centering
  \includegraphics[width=0.8\linewidth]{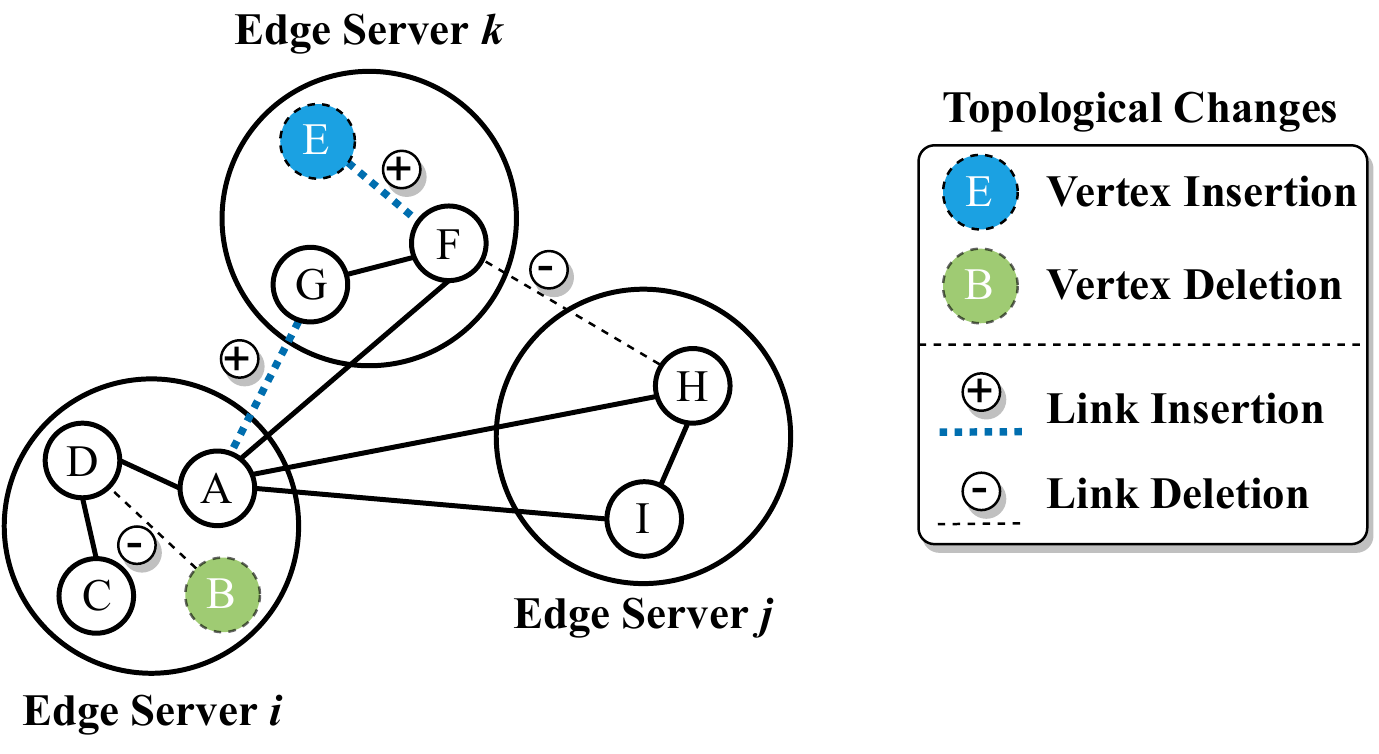}
  \caption{Possible topological graphs on a data graph, where the clients (vertices) have been placed to three edge servers. The changes include vertex insertion, vertex deletion, link insertion, and link deletion.
  }
  \label{fig:graph_dynamics}
\end{figure}

\subsection{Incremental Graph Layout Optimization}
Leveraging the insights in Sec. \ref{sec:offline_algorithm}, we propose to solve the dynamic issue with an incremental graph layout improvement strategy called GLAD-E (\textbf{GLAD} for \textbf{E}volved input graphs).
The key idea is to utilize the efficient transformation in GLAD-S, while limiting the scope of influenced vertices.

Before diving into algorithm details, we categorize graph evolution into four types from the perspective of cost minimization, corresponding to the instances in Fig. \ref{fig:graph_dynamics}.
1) \textit{Vertex insertion} on the existing graph layout clearly incurs additional cost (vertex \textit{E}), particularly for data collection, GNN computation, and edge server maintenance.
2) \textit{Vertex deletion} is instead positive to cost minimization obviously since it reduces the overall system workload (vertex \textit{B}).
3) \textit{Link insertion} has dual implications on the cost: if the terminal vertices of the new link are both located at the same server (link $\langle E,F \rangle$ in Fig. \ref{fig:graph_dynamics}),  the total cost remains changeless; otherwise a budget of traffic is demanded (link $\langle A,G \rangle$).
4) \textit{Link deletion} diminishes the total cost analogous to the vertex deletion case (links $\langle B,D \rangle$ and $\langle F,H \rangle$).
In summary, only a part of insertion cases needs consideration towards optimizing $\mathcal{P}$.

Keeping this fact in mind, we can largely narrow down the range of layout to be adjusted without hurting the total cost.
Algorithm \ref{algo:incremental} describes the procedure of GLAD-E.
Given the data graph at two successive time slots $t-1$ and $t$, we first filter those vertices that are involved in vertex insertion and cross-edge insertions.
Regarding the filtered ones, we can extract a subgraph $\mathcal{G}^+$ upon them from the present data graph $\mathcal{G}(t)$, and subsequently pass it as input to the GLAD-S algorithm and obtain an optimized layout $\pi^+$ (Line 2-3).
For the unfiltered vertices, we can also extract a graph layout $\pi^-$ of them from the existing layout $\pi(t-1)$.
Since $\pi^+$ and $\pi^-$ have accommodated all vertices at time slot $t$, incorporating them yields the optimized layout $\pi(t)$ adapting to the latest changes.
By doing so, each invocation of GLAD-E only operates on the vertices that possibly increase the total cost, effectively optimizing the system cost.
Moreover, as the filtered vertices only occupy a small portion of the whole, the overhead of incremental updates is minimum to the overall QoS.

\begin{algorithm}[t] 
\caption{GLAD-E: GLAD for Evolved input data graph} 
\label{algo:incremental} 
\begin{algorithmic}[1] 
    \REQUIRE ~~\\ 
    $\mathcal{T}$: The edge network\\
    $\mathcal{G}(t-1)$, $\mathcal{G}(t)$: The data graph at time slot $t-1$ and $t$\\
    $\pi(t-1)$: The graph layout at time slot $t-1$ \\
    $C(\pi)$: The total cost function for a given $\pi$\\
    $R$: The measurement of convergence\\
    \ENSURE ~~\\ 
    $\pi(t)$: Optimized graph layout at time slot $t$
    \STATE Filter the vertices that are newly added or have new neighbors at other edge servers from $\mathcal{G}(t-1)$ and $\mathcal{G}(t)$
    \STATE Construct a graph $\mathcal{G}^+$ with the filtered vertices and their associated links
    \STATE $\pi^+ \leftarrow $ GLAD-S($\mathcal{T}, \mathcal{G}^+, C, R$) \quad \quad \quad $\triangleright$ Call Algorithm \ref{algo:offline}
    \STATE Extract the graph layout $\pi^-$ from existing layout $\pi(t-1)$ with respect to the unfiltered vertices
    \STATE $\pi(t) \leftarrow \pi^+ \cup \pi^-$
    \RETURN $\pi(t)$
\end{algorithmic}
\end{algorithm}

\subsection{Adaptive Scheduling}

While the incremental strategy can achieve lightweight yet effective cost optimization, it can incur ``performance drifts" \cite{wang2020service} as time passes, since Algorithm \ref{algo:incremental} always checks for a local optimum over a part of the vertices.
The accumulation of performance loss over time may trap the system at a status where the advantage of incremental updates amortizes the penalty for migrating vertices.
In essence, this pertains to the tradeoff between the overhead and the benefit of a graph layout update. 
We therefore intend to find when such a status occurs and devise an adaptive scheduling algorithm for striking a balance on applying global optimization or incremental improvement.

Formally, given an existing graph layout $\pi(t-1)$ at time $t-1$ and an evolved data graph $\mathcal{G}(t)$ at time $t$, to update a layout $\pi(t)$ can either apply GLAD-S or GLAD-E.
Let $C^S(t)$ and $C^E(t)$ be the system cost achieved by these two ways\footnote{For the brevity of notations, we use $C^S(t)$ to simplify $C^S(\pi(t)|\mathcal{T}, \mathcal{G}(t))$, which means the system cost of the graph layout $\pi(t)$ generated by GLAD-S at time $t$. The same notation manner is used for GLAD-E.}, respectively, and the performance drift $f(t)$ is defined by the distance between them, \textit{i.e.} $|C^S(t)-C^E(t)|$.
Since GLAD-S commits to a convergence only if all edge pairs have been traversed at least once (\textit{cf.} Algorithm \ref{algo:offline} line 13), its searching space is obliged to be larger than that of GLAD-E.
This promises that the achieved cost of GLAD-S is not worse than that of GLAD-E, \textit{i.e.} $C^S(t) \leq C^E(t)$, and thus
\begin{align}
    f(t)=C^E(t)-C^S(t).
\end{align}

Assuming a threshold $\theta$ that quantifies the stringent cost budget required by Service Level Agreements (SLA), we recognize two statuses of the system: 1) when the accumulated performance drift $\sum_{t} f(t)$ meets SLA, \textit{i.e.} $\sum_{t} f(t) \leq \theta$, the service runs stable and we apply GLAD-E for incremental improvement, and 2) once $\sum_{t} f(t) > \theta$, the localized updates compensate for cost minimization and GLAD-S is supposed to be invoked.
The adaptive scheduling problem is to decide whether to switch the employed GLAD-E to GLAD-S at time $t$ such that the SLA $\theta$ is satisfied for a time as long as possible. 

A trivial solution for this problem is to record $f(t)$ for every time slot, and triggers GLAD-S as long as $\sum_{t} f(t) > \theta$.
However, at each time slot we can only invoke one algorithm for an update, meaning that it is infeasible to acquire $C^E(t)$ and $C^S(t)$ simultaneously, so for the value of $f(t)$.
Alternatively, we observe that $f(t)$ has an upper bound related to known parameters, as given in Theorem \ref{thm:bounded_performance_drift}.

\begin{algorithm}[t] 
\caption{GLAD-A: GLAD for Adaptive scheduling} 
\label{algo:scheduling} 
\begin{algorithmic}[1] 
    \REQUIRE ~~\\ 
    $\mathcal{G}(t)$: The data graph at time slot $t$ \\
    $\pi(t-1)$: The graph layout at time slot $t-1$\\
    $C(t-1)$: The system cost at time slot $t-1$ \\
    $f(t-1), f(t-2), \cdots, f(0)$: Historical performance drifts \\
    \ENSURE ~~\\ 
    $\Pi(t)$: The algorithm to be invoked at time slot $t$
    \STATE Approximate $f(t)$ according to Theorem \ref{thm:bounded_performance_drift}
    \IF{$\sum_t f(t) \leq \theta$}
        \STATE $\Pi(t) \leftarrow \ $ GLAD-E
    \ELSE
        \STATE $\Pi(t) \leftarrow \ $ GLAD-S
    \ENDIF
    \RETURN $\Pi(t)$
\end{algorithmic}
\end{algorithm}

\begin{theorem}[Bounded performance drift] \label{thm:bounded_performance_drift}
The performance drift $f(t)$ at time $t$ is upper bounded by $C(\pi(t-1)|\mathcal{G}(t))-C(t-1)$, where $C(\pi(t-1)|\mathcal{G}(t))$ is the cost of the unadjusted graph layout $\pi(t-1)$ over the evolved data graph $\mathcal{G}(t)$ and $C(t-1)$ is the system cost at time slot $t-1$.
\end{theorem}

\begin{proof}
We observe that $f(t)$ can be transformed into
\begin{align}
    f(t) = [C^E(t)-C(t-1)] - [C^S(t)-C(t-1)], \label{eq:transform_ft}
\end{align}
where $C(t-1)$ is the system cost at the last time slot $t-1$.
If the topological changes across time slots $t-1$ and $t$ are only vertex/link deletions, GLAD-E and GLAD-S will have the same operations on the graph layout, thus $f(t)=0$.
Otherwise, their adjustments will vary.
Ideally, we expect calling GLAD-S can accommodate all cost augmentation introduced by topological changes, \textit{i.e.} $C^S(t-1)-C(t-1)=0$ since it operates globally.
Meanwhile, calling GLAD-E will at least improve the existing graph layout $\pi(t-1)$, \textit{i.e.} $C^E(t) \leq C(\pi(t-1)|\mathcal{G}(t))$.
With the above analysis, we have
\begin{align}
    f(t) \leq C(\pi(t-1)|\mathcal{G}(t)) - C(t-1).
\end{align}

This completes the proof.
\end{proof}

Theorem \ref{thm:bounded_performance_drift} provides a computable bound on $f(t)$, where $C(t-1)$ is a known value after time $t-1$ and $C(\pi(t-1)|\mathcal{G}(t))$ can be reckoned once $\pi(t-1)$ and $\mathcal{G}(t)$ are provided.
Specifically, for the unchanged part, we compute their cost directly according to $\pi(t-1)$; 
for the deleted vertex/link, we omit them since they do not augment the system cost;
for an inserted link, its incremental cost can be easily calculated given its terminals' locations;
for an inserted vertex, we place it to the edge server that induces the maximum cost, in order to complement the upper bound.

Using the bounded result in Theorem \ref{thm:bounded_performance_drift}, we can estimate $\sum_t f(t)$ and devise an \textbf{A}daptive scheduling algorithm GLAD-A to decide the invocations of GLAD-E and GLAD-S.
Algorithm \ref{algo:scheduling} presents its pseudocode.
With the approximated $f(t)$ and recorded historical performance drifts, we can directly calculate $\sum_t f(t)$ and compare it with the SLA $\theta$, which steers to a routine of either GLAD-E or GLAD-S.
The time complexity of GLAD-A is dominated by the approximation of $f(t)$, which requires checking all topological changes on the data graph and calculating the system cost in complexity of $O(V)$.
Since other operations in Algorithm \ref{algo:scheduling} are in constant complexity, its total time complexity is $O(V)$.

\section{Evaluation}
\label{sec:evaluation}

This section presents the results of our evaluations using a variety of realistic settings, demonstrating the proposed GLAD solution can achieve significantly better results than existing approaches.

\subsection{Experimental Setup}
\label{sec:experimental_setup}

\begin{figure*}[t] 
    \centering
    \begin{minipage}[t]{0.23\textwidth}
        \centering
        \includegraphics[height=3.1cm]{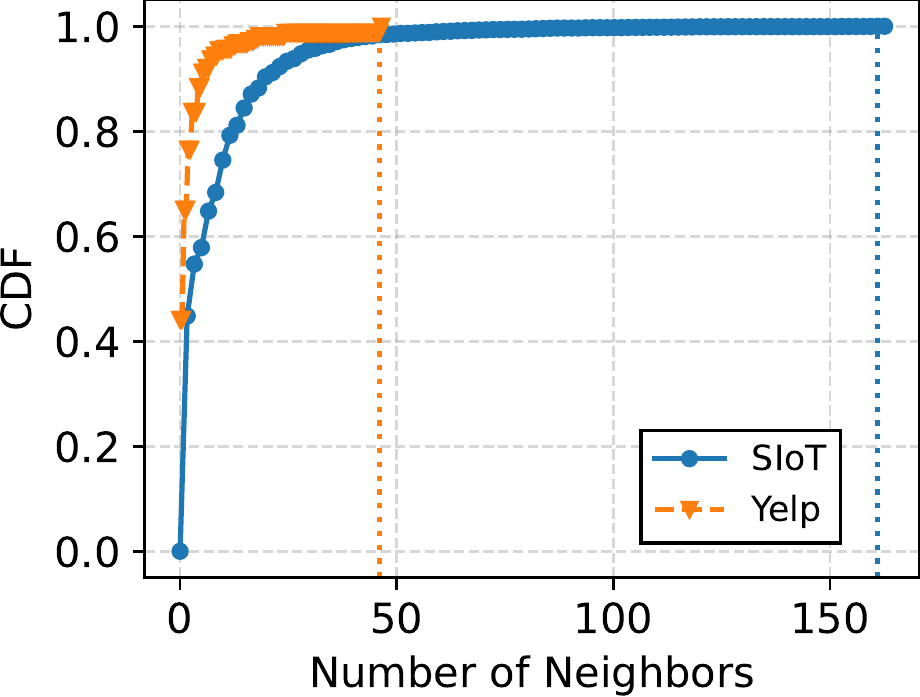}
        \caption{The CDF of the number of vertices' neighbors.}
        \label{fig:cdf_degree}
    \end{minipage}
    \quad
    \begin{minipage}[t]{0.23\textwidth}
        \centering
        \includegraphics[height=3.1cm]{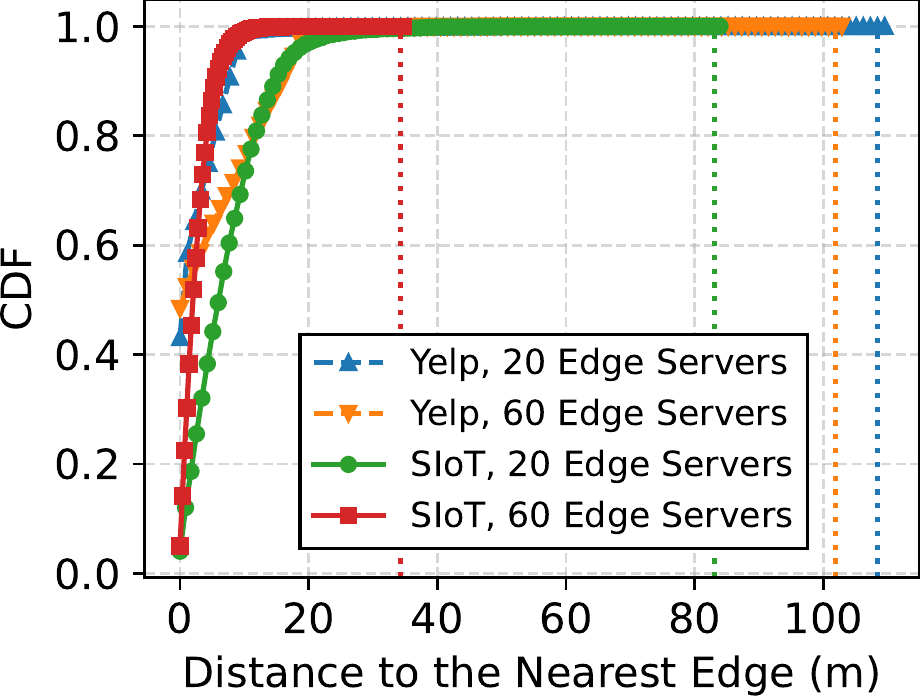}
        \caption{The CDF of the distance between clients to their nearest edges.}
        \label{fig:cdf_distance}
    \end{minipage}
    \quad
    \begin{minipage}[t]{0.23\textwidth}
        \centering
        \includegraphics[height=3.1cm]{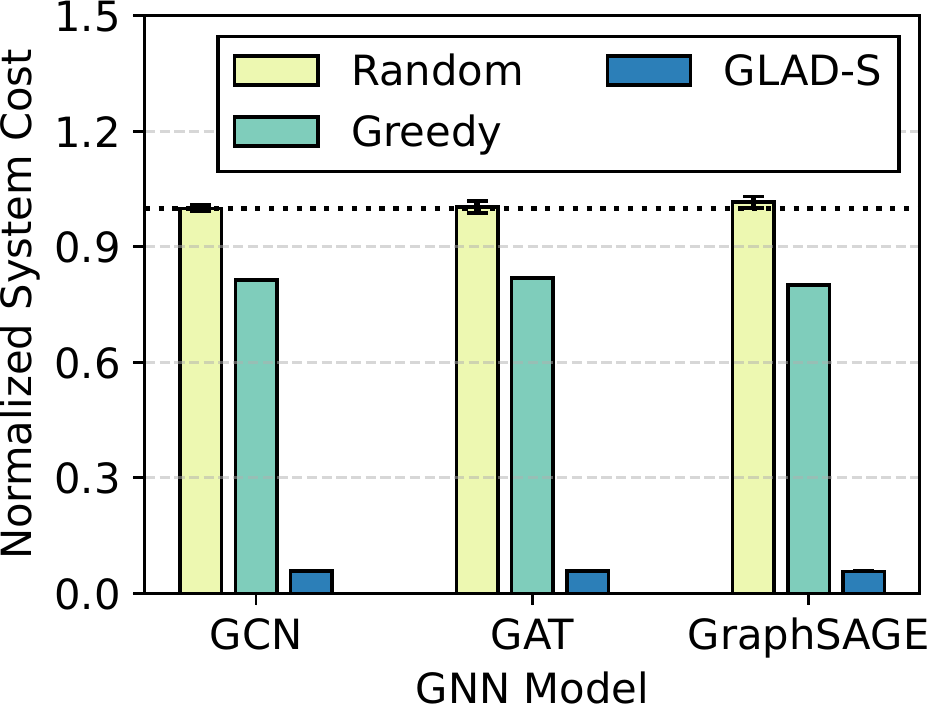}
        \caption{The achieved system cost of different GNN models on SIoT.}
        \label{fig:total_cost_siot}
    \end{minipage}
    \quad
    \begin{minipage}[t]{0.23\textwidth}
        \centering
        \includegraphics[height=3.1cm]{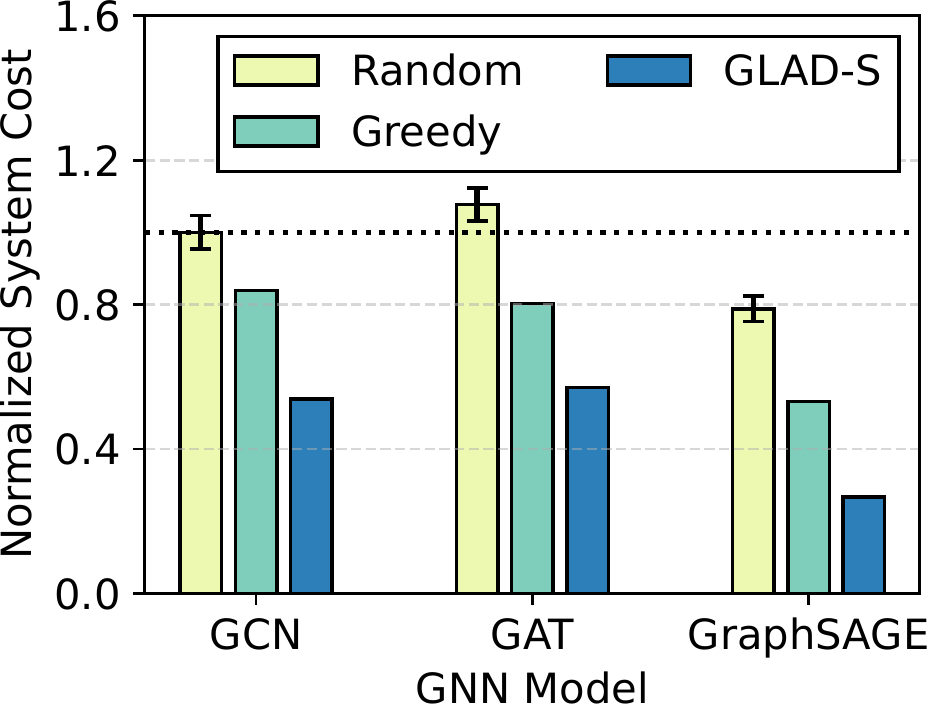}
        \caption{The achieved system cost of different GNN models on Yelp.}
        \label{fig:total_cost_yelp}
    \end{minipage}
\end{figure*}

\begin{table}[t]
    \caption{Specifications of used edge servers.}
    \label{tab:edge_servers}
    \centering
    \begin{tabular}{|c|c|c|c|} 
        \hline
        \textbf{Type} & \textbf{Processor} & \textbf{Mem.} &\textbf{Capability}  \\ \hline \hline
        \textit{A}    & 3.40GHz 8-Core Intel i7-6700 & 4GB    & Weak           \\ 
        \textit{B}    & 3.40GHz 8-Core Intel i7-6700 & 8GB    & Moderate       \\ 
        \textit{C}    & 3.70GHz 16-Core Xeon W-2145 & 32GB   & Powerful       \\
        \hline
    \end{tabular}
\end{table}

\textbf{Datasets.} 
Two realistic datasets are employed in our evaluation since they have been used in existing GNN-based services discussed in Sec. \ref{sec:gnn} and are the representative use cases with broad edge deployment \cite{zhong2020hybrid, dou2020enhancing, khanfor2020graph, jiang2022graph, gao2020deep}.
The first is SIoT \cite{marche2020exploit}, a graph of socially-coupled end devices collected in Santander, Spain. 
We randomly sample about half of the vertices from the complete dataset and obtain a graph consisting of 8001 vertices and 33509 links.
Each vertex has a 52-dimension feature that encodes its properties such as the device’s type, brand, and mobility, and is labeled as public/private depending on the corresponding device's physical ownership (by a person or an organization).
The GNN service over SIoT is to recognize the labels of clients, by predicting from their features and relationships \cite{hamrouni2022service, aljubairy2021towards}.

The second is Yelp, a graph randomly extracted from its complete backup \cite{rayana2015collective, yelpchi}, which collects 3912 customers' reviews with 4677 connections for a pool of hotels and restaurants in Chicago.
A vertex in Yelp is a client's review comment, represented by a 100-dimension Word2Vec \cite{mikolov2013efficient} vector, and a link between two vertices means the two reviews are posted by the same user.
The reviews can be categorized into two types: spam comments that are fake and need filtering, and normal comments that should be reserved.
Their labels are obtained by Yelp's official review filters \cite{mukherjee2013yelp} and are provided in the original dataset.
The GNN service for Yelp is to identify the spammers, according to the graph-based fraud detection applications \cite{dou2020enhancing, liu2020alleviating}.
Fig. \ref{fig:cdf_degree} depicts the Cumulative Distribution Function (CDF) of the number of vertices' neighbors in the graphs.
In SIoT, the degree distribution exhibits a long-tail property, whereas Yelp is more sparse and has many isolated vertices.
This indicates the used graphs are in distinct characteristics and we will show our approach can perform well in both datasets.

To obtain the locations of edge servers, we follow the settings in \cite{wang2019edge} and perform k-means \cite{lloyd1982least} on the clients' spatial coordinates to obtain clustering pivots, which act as the edge server locations.
Since the Yelp dataset is absent in clients' geographical locations, we borrow the idea from ``workload composition" \cite{kolosov2020benchmarking} and use the positional information of New York taxis \cite{NYC_taxi} for synthesis.
Fig. \ref{fig:cdf_distance} displays the CDF of the distance between clients to their nearest edge servers, when the number of edge servers is set as 20 and 60, respectively.
The distributions for SIoT and Yelp are different: some clients in Yelp are solitary and far from the edge servers, meaning a potentially costly data uploading process; clients in SIoT are more compact to access the edge network.
Besides, the larger the number of edge servers, the smaller the average access distance is.

\textbf{Parameters.} Our experiments are carried out with GCN \cite{kipf2016semi}, GAT \cite{velivckovic2017graph}, and GraphSAGE \cite{hamilton2017inductive}, three representative GNN models that have been widely adopted in edge scenarios \cite{zhong2020hybrid, dou2020enhancing, guo2019attention}.
Their formal definitions of layer-wise execution semantics have been explained in Sec. \ref{sec:gnn}.
All the models are implemented in two layers using the instances from Pytorch Geometric \cite{fey2019fast} model zoo and are trained prior to deployment.
The input feature size of the models are 52 and 100 for SIoT and Yelp, respectively, while the number of hidden units is fixed at 16 and the output size is 2 (binary classification).
For heterogeneous edge resources, we deploy three types of edge servers, labeled as type \textit{A} (weak), \textit{B} (moderate), and \textit{C} (powerful), respectively, and their specifications are listed in Table \ref{tab:edge_servers}.
To gauge the parameters in computation cost $C_P$, \textit{i.e.} $\alpha_i$, $\beta_i$, and $\gamma_i$, we profile the operator-wise execution time on all types of machines by passing GNN models through the above-mentioned datasets.
All the background loads are turned off during the profiling and each benchmark is repeated 50 times to obtain stable results.
For the unit cost of data collection $\mu_{vi}$, we calculate them by multiplying a factor with the geographical distances between client $v$ and edge $i$ \cite{jiao2014multi}, also the same way is applied to obtain $\tau_{ij}$ for cross-edge traffic at edge pair $\langle i , j \rangle$.
For the maintenance cost, we generate $\rho_{vi}$ and $\epsilon_i$ using a Gaussian process, following the fact that the hourly electricity prices in some US cities obey the Gaussian distribution \cite{Qureshi2009cutting}.
By default, $R$ is set to 3, representing that a valid convergence is checked when the cost cannot be further reduced in three successive searching attempts, while we will also inspect the effect of varying $R$ in our evaluation.

\begin{figure*}[t] 
    \centering
    \begin{minipage}[t]{0.23\textwidth}
        \centering
        \includegraphics[height=3.0cm]{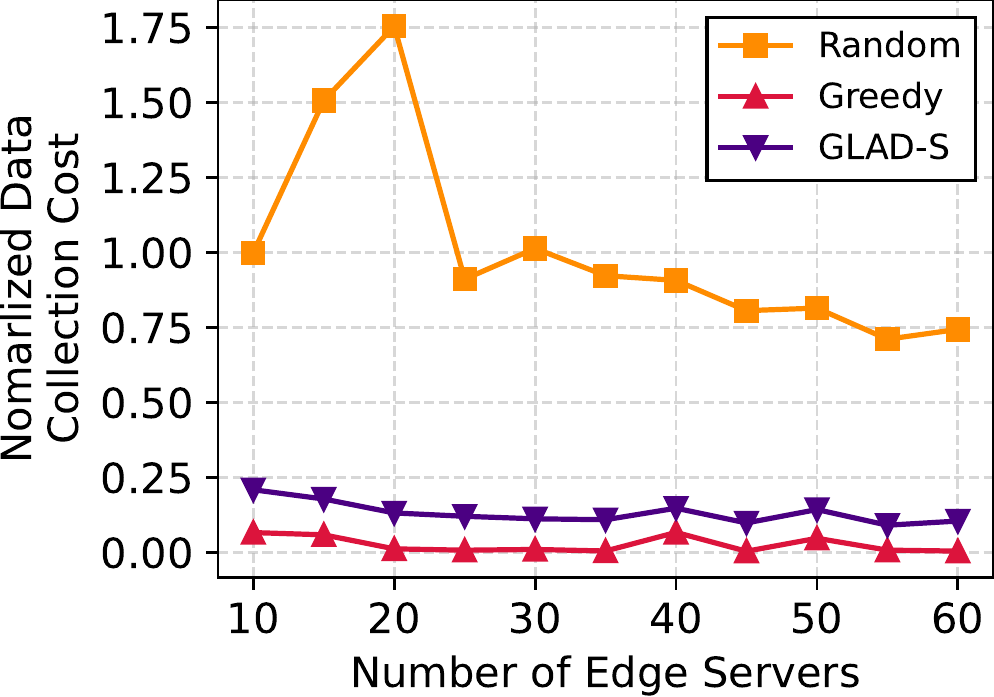}
        \caption{Data collection cost on Yelp with varying number of edges.}
        \label{fig:uploading_cost}
    \end{minipage}
    \quad
    \begin{minipage}[t]{0.23\textwidth}
        \centering
        \includegraphics[height=3.0cm]{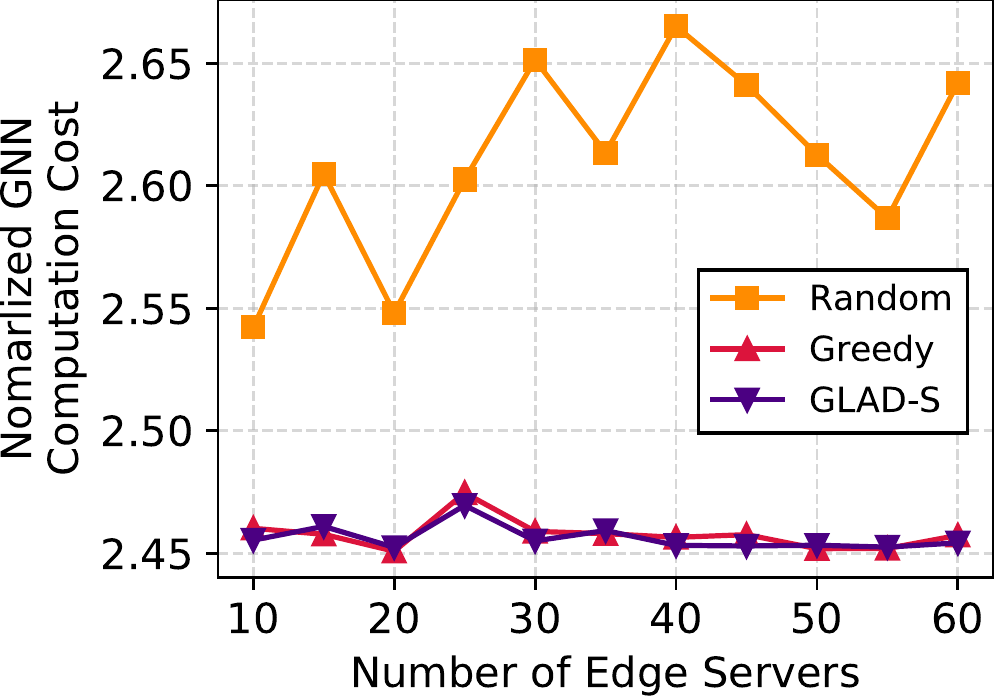}
        \caption{GNN computation cost on Yelp with varying number of edges.}
        \label{fig:computation_cost}
    \end{minipage}
    \quad
    \begin{minipage}[t]{0.23\textwidth}
        \centering
        \includegraphics[height=3.0cm]{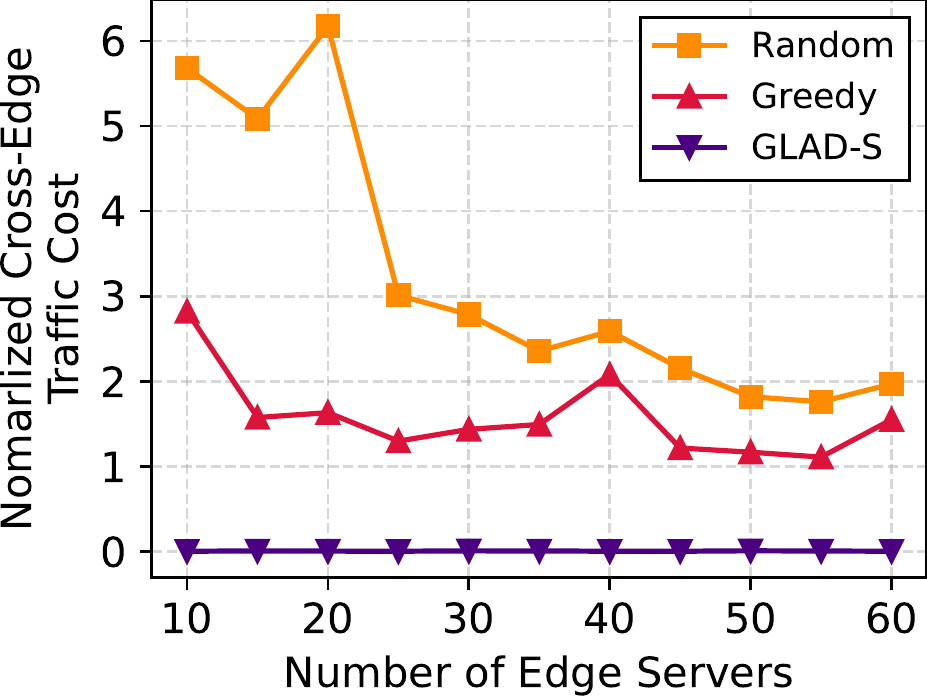}as
        \caption{Cross-edge traffic cost on Yelp with varying number of edges.}
        \label{fig:traffic_cost}
    \end{minipage}
    \quad
    \begin{minipage}[t]{0.23\textwidth}
        \centering
        \includegraphics[height=3.0cm]{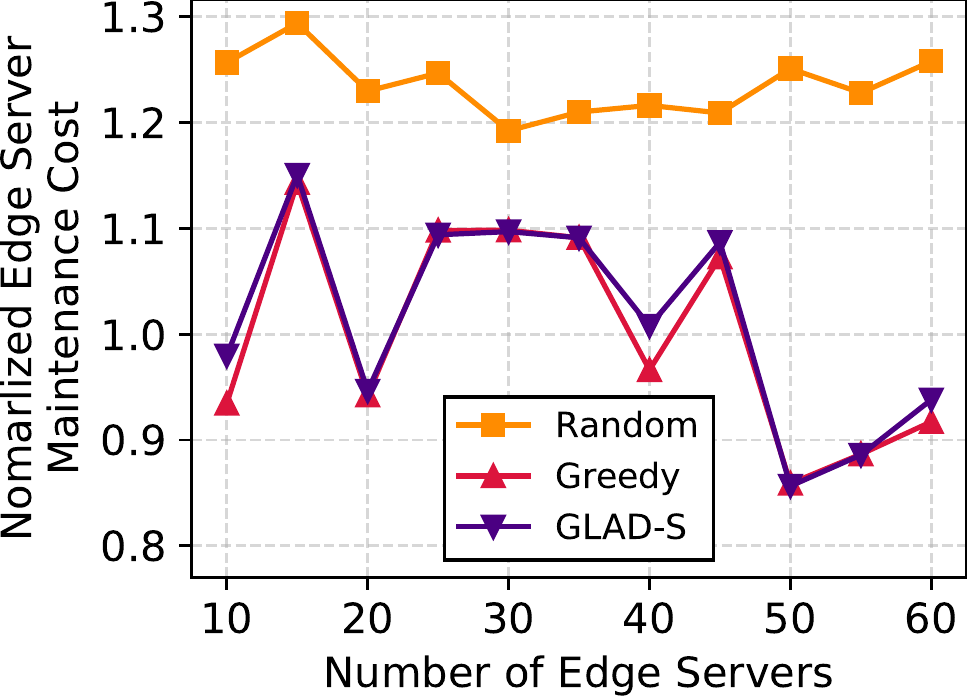}
        \caption{Maintenance cost on Yelp with varying number of edges.}
        \label{fig:maintenance_cost}
    \end{minipage}
\end{figure*}

\textbf{Methodology.}
We implement GLAD solutions with a state-of-the-art max-flow algorithm \cite{Orlin2013max} to solve the \textit{s-t} cut series.
All the experimental results are recorded as the mean value of 20 repetitions.
To emulate heterogeneity, we label each edge server in equal proportion as one of the three types of machines in Table \ref{tab:edge_servers}.
Specifically, if the number of edge servers cannot be divided into three groups exactly, we assign the remainders in a priority of $A$, $B$, and $C$.
For instance, if we simulate twenty edge servers, the system will involve seven machines of type $A$, seven of $B$, and six of $C$.
GLAD does not touch upon the accuracy issues as the execution semantics of GNN models are not modified, and the raw graph data as well as the model parameters are fully reserved.
Hence, the model accuracy only depends on the input GNN model itself, and is irrelevant to our proposed cost-optimized graph layout scheduling solution.
Our evaluation thus concentrates particularly on the cost metric, including the total cost and detailed cost factors, and conducts experiments for static and dynamic input graphs, respectively.

In the static setting, we compare GLAD-S with two baselines: 1) \textit{Random}, which generates a graph layout with each client assigned to an arbitrary edge server, and 2) \textit{Greedy}, which places clients to the edge servers that greedily minimizes the cost of data collection, GNN computation, and edge server maintenance.
These two heuristics are yet widely applied in many data placement solutions \cite{udayakumaran2006dynamic, azizi2020grvmp, xie2019efficient}.
Traditional data placement techniques for CNN and RNN are not qualified as baselines because these models' input data (\textit{e.g.} image, voice signal) are typically from a single source (camera, microphone) for individual queries, whereas GNN's input is graph data distributed among multiple devices.
FL-oriented methods are also beyond the scope since they aim at a privacy-preserving distributed training with centralized aggregators (usually cloud servers), while our considered DGPE conceives a distributed inference with graph data availably shared among distributed edge servers.
Such substantial distinction engenders entirely different optimization principles for their data placement strategies, and thus these methods and GLAD are not at the same comparable level.

In the dynamic setting, we focus on examining our proposed approaches against graph evolution in a time window.
To generate traces of topological graph changes, for each time slot we first select a percentage of changed links and calculate the corresponding number of links based on $|\mathcal{E}|$, \textit{i.e.} the amount of links in the input data graph.
Next, we produce an integer following the Gaussian distribution with the link number obtained above as the mean value and a half of that as the standard deviation.
Finally, using this integer as a beacon, we randomly select a number of vertices in the data graph, and uniformly generate an insertion or deletion of links between them.
The same way is applied for generating a number of vertex insertions/deletions.
By doing so, we can obtain a series of topological changes on vertices and links that simulates graph evolution.
To demonstrate the performance of GLAD-E and how adaptive GLAD-A can be, for each time slot we compare four approaches: 1) \textit{No Adjustment}, where neither algorithm is invoked during the runtime, 2) \textit{Greedy}, which utilizes the greedy approach in static graph setting for online graph layout optimization, 3) \textit{GLAD-E}, which calls GLAD-E for incremental graph layout scheduling for every time slot, and 4) \textit{Adaptive} that applies GLAD-A to selectively decide using either GLAD-E or GLAD-S.

\subsection{Performance Comparison}

We first compare the proposed GLAD-S with Random and Greedy baselines by investigating the total system cost with static input graphs.
We select 60 edge servers, with the three types of machines in equal proportion.
Fig. \ref{fig:total_cost_siot} shows the normalized system cost obtained by these approaches on different GNN models over the SIoT dataset.
As we witness, GLAD consistently outperforms Random and Greedy, achieving at most 94.1\%, 94.4\%, and 95.8\% cost-saving for GCN, GAT, and GraphSAGE, respectively.
Traversing horizontally, GraphSAGE takes a lower cost budget than that of GCN and GAT whichever approach is applied, meaning that GraphSAGE requires a lighter system workload.
The results for Yelp in Fig. \ref{fig:total_cost_yelp} are analogous to that in SIoT (Fig. \ref{fig:total_cost_siot}), where GLAD always produces the smallest system cost, demonstrating that the proposed approach generalizes well and can handle different types of input graphs and GNN models.

We take one step closer to explore how each cost factor varies with different numbers of edge servers.
Specifically, we run GAT over Yelp, and record the costs of data collection, GNN computation, cross-edge traffic, and edge server maintenance, for different approaches.
To compare the contribution of different cost factors, all results are normalized based on the value of Random's data collection cost with 10 edge servers.
Fig. \ref{fig:uploading_cost} shows the results for data collection cost $C_U$.
Random and Greedy exhibit the ceiling and floor for all cases, respectively.
This is because the graph layout in Random is irrespective of clients' spatial distribution, while Greedy always prefers the physically nearest edge servers for clients.
GLAD-S's results locate between the other two baselines, and are much smaller than Random's while closer to Greedy's.
As the number of edge servers increases, the geographical distances between clients and edges diminish, and hence the data collection cost of GLAD gradually shrinks.
Fig. \ref{fig:computation_cost} and Fig. \ref{fig:maintenance_cost} respectively describe the GNN computation cost $C_P$ and the edge server maintenance cost $C_M$, where GLAD-S takes clearly smaller costs than Random.
This accredits GLAD's heterogeneity awareness, and their performance gap implicates the optimization effect.
Concretely, among a cluster of edge servers, those with higher computing capability allow lower processing overhead under the same workload level, and therefore will be assigned with more vertices such that the total system cost is retrenched.
The Greedy baseline's optimization also takes computation cost into consideration and therefore displays similar records as GLAD-S.
Nevertheless, their curves are not fully overlapped, since GLAD-S's objective further accommodates the traffic cost between edge servers and needs to jointly reconcile all cost factors.
As shown in Fig. \ref{fig:traffic_cost}, GLAD-S significantly outperforms the other two baselines by lying at the bottom in the dimension of cross-edge traffic cost $C_T$.
The much smaller values imply that GLAD-S places the data graph's links to a small number of edge servers, while alleviating inter-server transmission as much as possible.
Across all the four figures, we note that the cross-edge traffic cost contributes a majority of the total system cost, for which GLAD-S' optimization accomplishes superior performance improvement over other baselines.

\begin{figure}[t] 
    \centering
    \begin{minipage}[t]{0.23\textwidth}
        \centering
        \includegraphics[height=3.0cm]{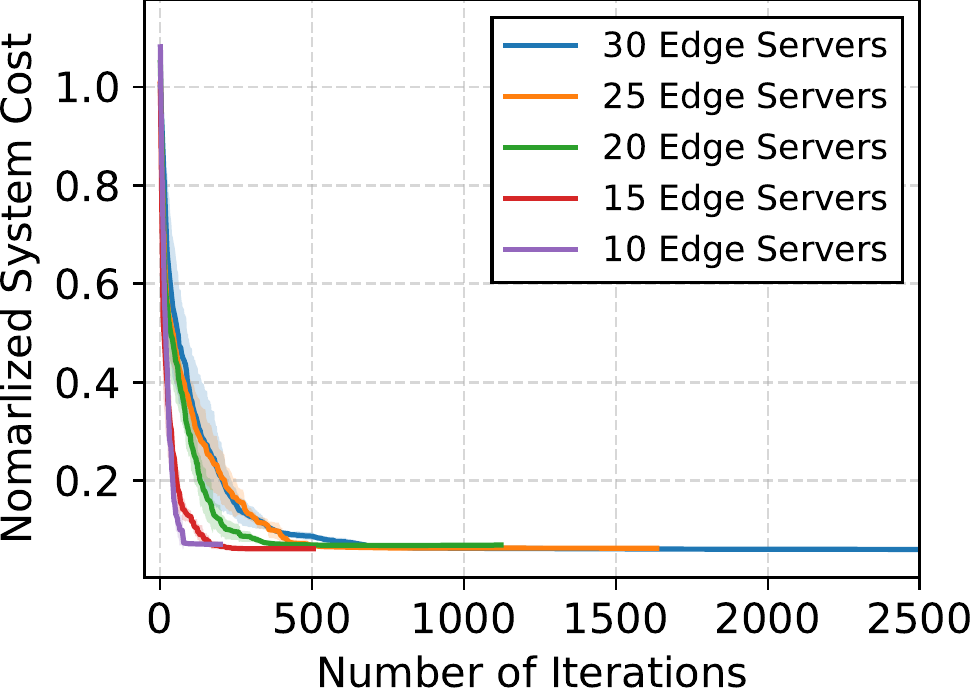}
        \caption{The cost convergence of GLAD-S on SIoT.}
        \label{fig:convergence_siot}
    \end{minipage}
    \quad
    \begin{minipage}[t]{0.23\textwidth}
        \centering
        \includegraphics[height=3.0cm]{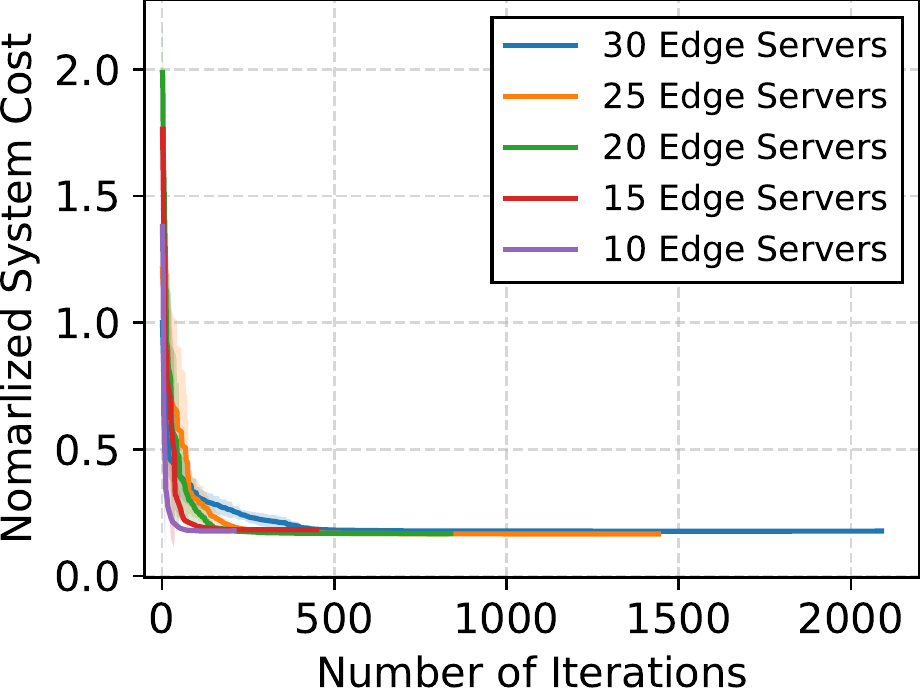}
        \caption{The cost convergence of GLAD-S on Yelp.}
        \label{fig:convergence_yelp}
    \end{minipage}
\end{figure}

\subsection{Convergence}
Fig. \ref{fig:convergence_siot} and Fig. \ref{fig:convergence_yelp} illustrate the system costs after each iteration of GLAD-S, varying the number of edge servers in DGPE.
GraphSAGE is employed to infer over SIoT and Yelp, respectively.
From the results we have three observations.
First, as iterations go ahead, all curves continuously decrease the total costs, pushing system performance to sweeter positions.
The effectiveness and convergence of GLAD-S are thus corroborated, given any number of edge servers.
Second, among all cases for both SIoT and Yelp, we find that the algorithm converges fast, where the system cost declines exponentially.
Third, the most cost diminution is achieved in the first few iterations, especially for Yelp in Fig. \ref{fig:convergence_yelp}.
The performance improvement shows a marginally decrement effect, in accordance with the objective function’s submodularity.
This encourages an early-exit strategy in practical deployment, where one can even terminate the algorithm ahead of actual convergence to cut down the scheduling time, since the optimization effect of the first few iterations is yet sufficient for some applications.

\subsection{Adaptability}
\label{sec:adaptability}
Fig. \ref{fig:adaptability} depicts the results for dynamic graph evolution using GAT over Yelp in 200 time slots.
With 10 edge servers employed, we investigate four approaches: No Adjustment, Greedy, GLAD-E, and Adaptive.
We perform GLAD-S to obtain an initial graph layout, set the percentage of topological changes within the whole graph as 1\%, and generate link insertion/deletion at each time slot following the method in Sec. \ref{sec:experimental_setup}.
As recorded in the top subfigure, the number of links in the input data graph continuously fluctuates over time.
Referring to the original number of links at the initial moment (the dashed line), the graph size changes in a relatively small magnitude, which aligns with the real evolution in GNN's input data graph.
Following such evolution, the middle subfigure plots the decisions of our proposed GLAD-A algorithm on calling either GLAD-E or GLAD-S.
We observe from the invocations that GLAD-S is only applied for a few time slots, which significantly saves the scheduling overhead of global layout adjustment.
The bottom subfigure shows the achieved system cost during the time window for all approaches, where the recorded values are normalized with respect to the system cost at the initial time slot.
As we can see, both GLAD-S and GLAD-E consistently outperform No Adjustment and Greedy, confirming the benefit brought by graph layout optimization.
However, we also find an evident performance drift between their trajectories, where GLAD-E gradually compensate for its advantages in agile adjustment.
GLAD-A bridges this gap by dynamically switching between incremental updates and global optimization, and therefore achieving lower costs than the way applying GLAD-E only.
Overall, GLAD-E can remarkably lower the system cost during the runtime and GLAD-A can further refine the performance by adapting graph layout optimization agilely to dynamic graph evolution.

\begin{figure}[t]
    \centering
    \includegraphics[width=0.95\linewidth]{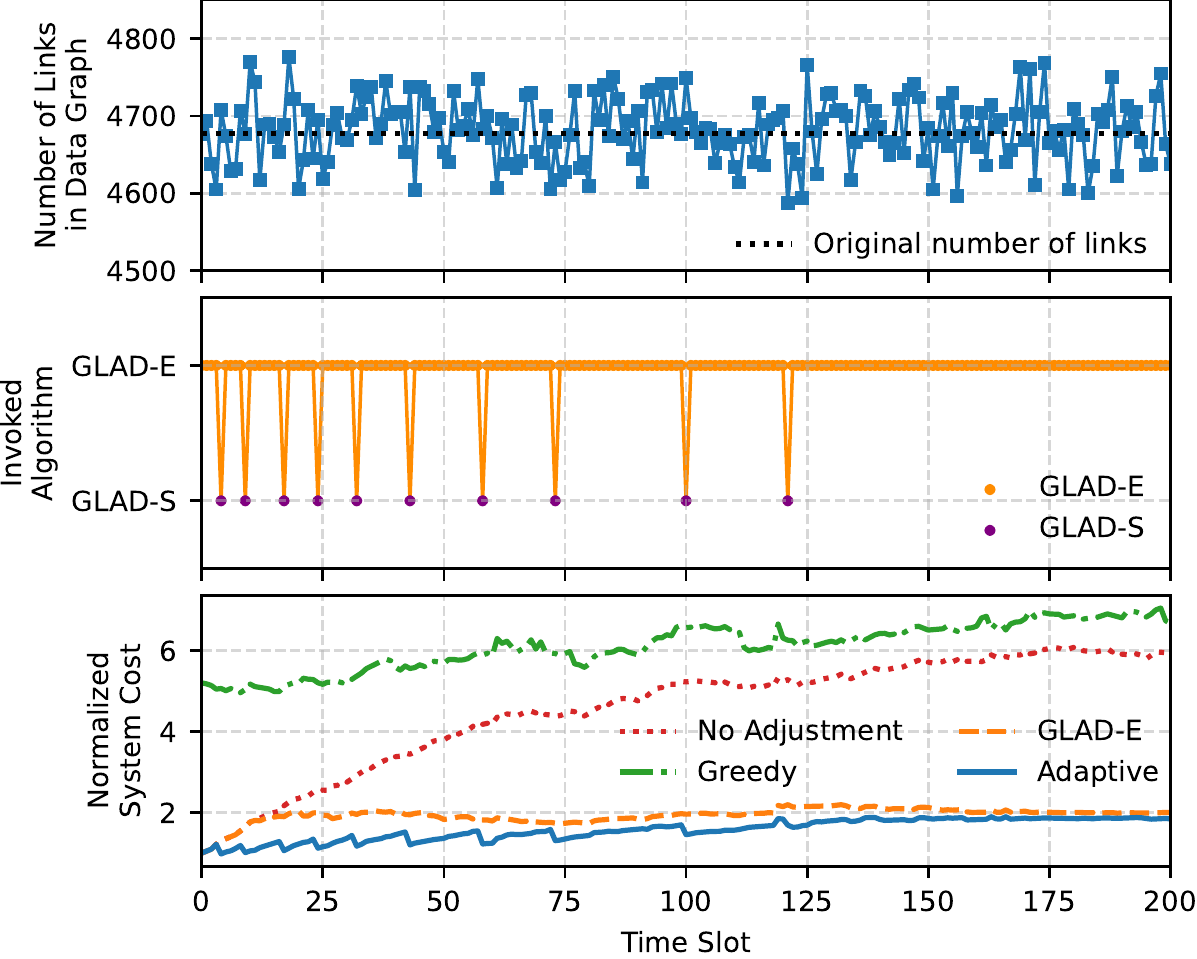}
    \caption{Performance fluctuation of different approaches with graph evolution. The proposed adaptive approach can dynamically adjust the invocations of GLAD-E and GLAD-S for lower system costs.}
    \label{fig:adaptability}
\end{figure}

\begin{figure*}[t] 
    \centering
    \begin{minipage}[t]{0.23\textwidth}
        \centering
        \includegraphics[height=3.1cm]{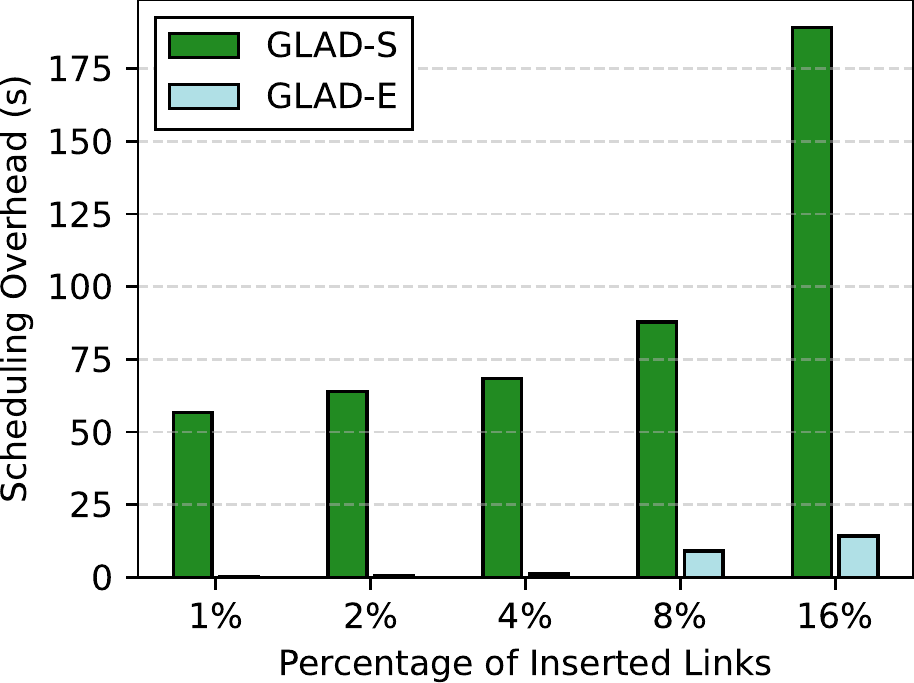}
        \caption{Scheduling overhead of GLAD-S and GLAD-E with varying link insertions on SIoT.}
        \label{fig:overhead_siot}
    \end{minipage}
    \quad
    \begin{minipage}[t]{0.23\textwidth}
        \centering
        \includegraphics[height=3.1cm]{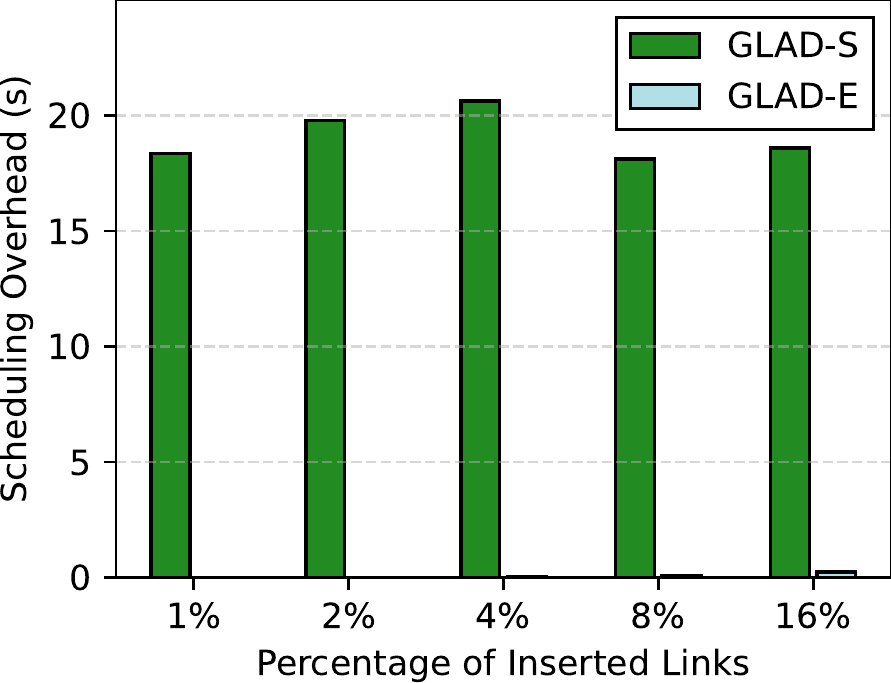}
        \caption{Scheduling overhead of GLAD-S and GLAD-E  with varying link insertions on Yelp.}
        \label{fig:overhead_yelp}
    \end{minipage}
    \quad
        \begin{minipage}[t]{0.23\textwidth}
        \centering
        \includegraphics[height=3.1cm]{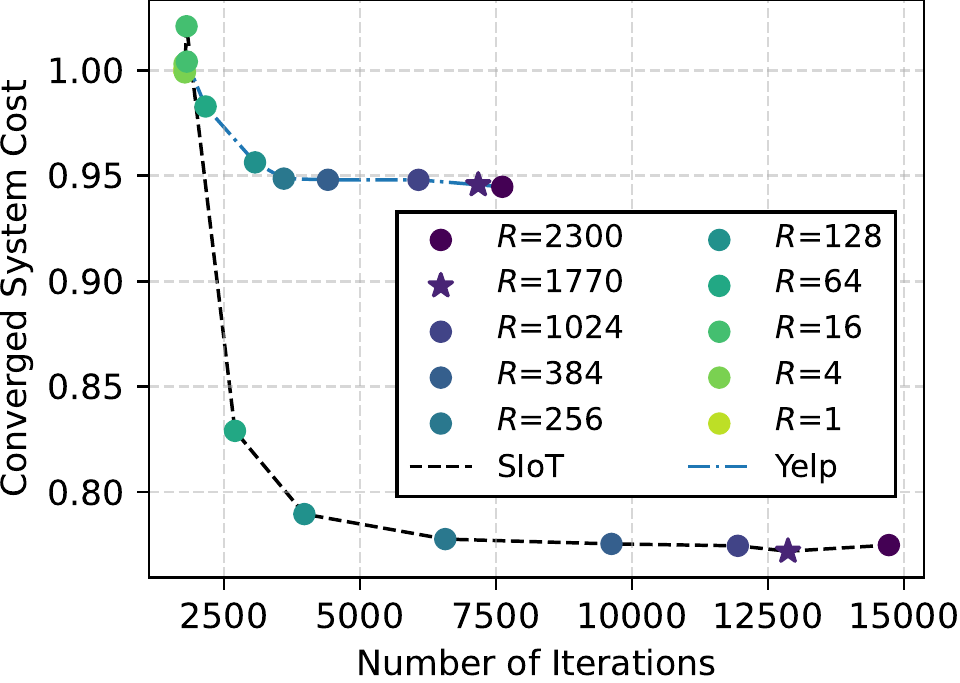}
        \caption{Converged system cost and the number of iterations of GLAD-S with varying $R$.}
        \label{fig:sensitivity_r}
    \end{minipage}
    \quad
    \begin{minipage}[t]{0.23\textwidth}
        \centering
        \includegraphics[height=3.1cm]{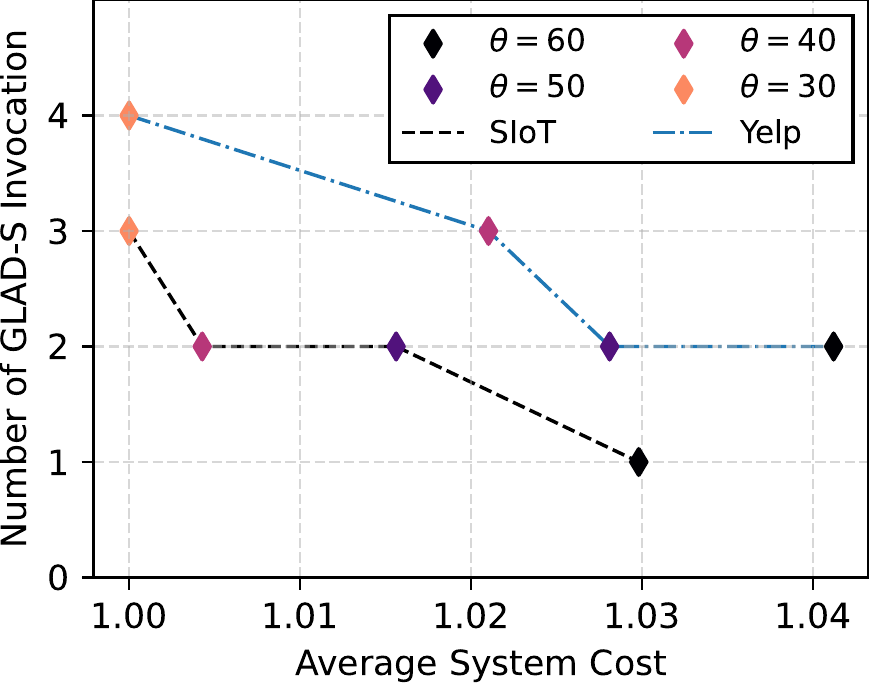}
        \caption{Average system cost in 200 time slots and the number of GLAD-S invocations with varying $\theta$.}
        \label{fig:sensitivity_theta}
    \end{minipage}
\end{figure*}

\subsection{Scheduling Overhead}
We further examine the scheduling overhead of GLAD-S and GLAD-E under the same experimental setup in Sec. \ref{sec:adaptability}.
Specifically, we vary the percentage of inserted links, generate a corresponding number of links to the graph, and record the scheduling time of two algorithms.
Fig. \ref{fig:overhead_siot} and Fig. \ref{fig:overhead_yelp} show the measurements for SIoT and Yelp, respectively.
In all cases, GLAD-E takes a much smaller overhead than GLAD-S, indicating a significant reduction in graph cut operations in incremental graph layout optimization.
Especially for Yelp in Fig. \ref{fig:overhead_yelp}, GLAD-E needs at most 0.3s (for 16\% link insertion) to finish a graph layout optimization.
Combining the results in Fig. \ref{fig:adaptability} evidences the proposed GLAD-E's lightweight property and scheduling effectiveness.
As the number of inserted links increases, GLAD-E spends gradually higher time for both two datasets, but always achieves a $<$10\% overhead compared to GLAD-S, which evidences the lightweight property of incremental improvement.
Nonetheless, given that the time slot for topological changes in the data graph is usually at the scale of minutes, all the algorithms in the GLAD family are practical for realistic deployment.

\subsection{Sensitivity}
\label{sec:sensitivity}
The last experiment studies the sensitivity of parameters ($R$ and $\theta$) used in our algorithms.
We first focus on $R$ in GLAD-S, running GAT over 60 edge servers under static graph settings.
Fig. \ref{fig:sensitivity_r} records the converged system cost and the number of experienced iterations of GLAD-S with varying $R$ for two datasets.
As $R$ increases, the system cost of the graph layout generated by GLAD-S first rapidly decreases and then gradually converges.
This attributes that a larger $R$ claims more convergence checking rounds across iterative graph cuts, which leads to a larger searching space and hence finds a better graph layout result.
In particular, when $R=\frac{|\mathcal{D}|(|\mathcal{D}|-1)}{2}$ (marked as a star in Fig. \ref{fig:sensitivity_r}), the achieved system cost can reach the converged value, suggesting that such configuration can steer to the local optimum.
Nevertheless, the gain in system cost improvement comes at a price.
Traversing the horizontal axis in Fig. \ref{fig:sensitivity_r} reveals that magnifying $R$ means much more iteration times, which gives rise to longer scheduling time.
This implies that the selection of $R$ tightly relates to the tradeoff between the quality of obtained graph layout and scheduling overhead.
In practical deployment, we can strike a balance between them by configuring a moderate value for $R$, \textit{e.g.} $R=256$ for SIoT, which can attain a majority of cost reduction with relatively fewer iterations.

We next investigate the impact of $\theta$ in GLAD-A, under dynamic graph settings. 
Fig. \ref{fig:sensitivity_theta} shows the number of GLAD-S invocations and the average system cost in 200 time slots with varying $\theta$ for two datasets.
As $\theta$ increases, the tolerance on QoS expands, and hence calling GLAD-S for global graph layout adjustment appears fewer times.
For instance, in SIoT, the adaptive strategy with $\theta=60$ will only invoke GLAD-S once within the time window, indicating that it resorts to GLAD-E and update the graph layout in an incremental manner for almost all time slots.
From the axis of average system cost, a smaller $\theta$ clearly leads to a smaller result, demonstrating the benefit of global layout optimization (using GLAD-S).
Analogous to $R$, the configuration of $\theta$ also involves the overhead-outcome tradeoff, since the number of GLAD-S invocations directly concerns the scheduling time.
Network operators can accordingly regulate $\theta$ to tune the preference between system cost and scheduling overhead.

\section{Conclusion and Future Work}
\label{sec:conclusion}

In this paper, we build a novel modeling framework to characterize cost optimization in edge-enabled distributed GNN processing, which is generalizable and compatible with a wide range of objectives.
Leveraging the properties of the system cost function, we derive an effective graph-cut based algorithm with theoretical performance guaranteed.
We further develop an incremental graph layout improvement strategy and an adaptive scheduling algorithm to achieve lightweight yet efficient cost minimization under dynamic graph evolution.
Extensive evaluations corroborate the superiority of the proposed approach, as well as its convergence and adaptability.
As GNNs have received growing attention and adoption in edge computing scenarios, the proposed framework can serve as a foundation for future analysis and optimization on specific GNN-driven applications.
In the future, we will further study extended DGPE cases where clients can replicate their vertices data across multiple edge servers, and investigate privacy-preserved distributed GNN training under the FL paradigm.

\appendices

\section{Proof of Theorem \ref{thm:correctness}}
\begin{proof}
We prove by contradiction.
Assuming the cost obtained by minimum \textit{s-t} cut on $\mathcal{A}(i,j)$ is $c$ and the minimum cost for the clients and edge servers in $\mathcal{A}(i,j)$ is $c^*$, we have $c > c^*$.
According to the mapping policy in Eq. (\ref{eq:mapping}), we can use the optimal placement decision of $c^*$ to construct a cut set in $\mathcal{A}(i,j)$, which is supposed to be the set of edge cuts with minimum weights.
This subsequently contradicts the fact that $c$ is the cost of the minimum \textit{s-t} cut on $\mathcal{A}(i,j)$.
Hence, $c = c^*$ and Theorem \ref{thm:correctness} follows.
\end{proof}

\section{Proof of Theorem \ref{thm:approximation}}

\begin{proof}
Select an arbitrary edge server $i \in \mathcal{D}$ and define the set of clients that are assigned to $i$ in the global optimal placement $\pi^*$ as $\mathcal{V}_i$.
From the placement $\pi$ produced by GLAD-S, we generate a new placement $\pi'$ by assigning the clients in $\mathcal{V}_i$ to edge server $i$ and keep other clients located at their original edges.
Since $\pi$ is the optimal solution within GLAD-S's searching space, switching from $\pi$ to $\pi'$ will increase the total cost, \textit{i.e.}
\begin{align}
    C(\pi) \leq C(\pi'). \label{eq:cost_local_optimum}
\end{align}

We now group all links in $\mathcal{G}$ into three sets in terms of $\mathcal{V}_i$, namely the internal set $\mathcal{E}^{\text{in}}_i$, external set $\mathcal{E}^{\text{ex}}_i$, and the cross set $\mathcal{E}^{\text{cr}}_i$:
\begin{align}
    \mathcal{E}^{\text{in}}_i &= \{\langle u, v\rangle \in \mathcal{E}| u\in \mathcal{V}_i, v \in \mathcal{V}_i\}, \\ 
    \mathcal{E}^{\text{ex}}_i &= \{\langle u, v\rangle \in \mathcal{E}| u\notin \mathcal{V}_i, v \notin \mathcal{V}_i\}, \\
    \mathcal{E}^{\text{cr}}_i &= \{\langle u, v\rangle \in \mathcal{E}| u\in \mathcal{V}_i, v \notin \mathcal{V}_i\}.
\end{align}

This consequently partition the data graph $\mathcal{G}$ into three subgraphs:
\begin{align}
    \mathcal{G}^{\text{in}}_i &= (\mathcal{V}_i, \mathcal{E}^{\text{in}}_i), \\
    \mathcal{G}^{\text{ex}}_i &= (\mathcal{V} \setminus \mathcal{V}_i,  \mathcal{E}^{\text{ex}}_i), \\
    \mathcal{G}^{\text{cr}}_i &= (\emptyset, \mathcal{E}^{\text{cr}}_i).
\end{align}

Recall that in Eq. (\ref{eq:cost_three}) the total cost can be divided into three components, where $C_2$ is the quadratic term that depends on the placement of client pairs, $C_1$ is the linear term with respect to the decision for each client, and $C_0$ is the one-shot cost for all edge servers.
In particular, we are interested in $C_1$ and $C_2$ since they are related to the decision variables, and we denote $C^{+} = C_1 + C_2$ for the brevity of notations.

Centering around the three subgraphs above, we highlight the following relations with respect to $C^{+}$:
\begin{align}
    C^{+}(\pi'|\mathcal{G}^{\text{in}}_i) &= C^{+}(\pi^*|\mathcal{G}^{\text{in}}_i), \label{eq:relation_internal} \\
    C^{+}(\pi'|\mathcal{G}^{\text{ex}}_i) &= C^{+}(\pi|\mathcal{G}^{\text{ex}}_i), \label{eq:relation_external}\\
    C^{+}(\pi'|\mathcal{G}^{\text{cr}}_i) &\leq \lambda C^{+}(\pi^*|\mathcal{G}^{\text{cr}}_i). \label{eq:relation_cross}
\end{align}
Eq. (\ref{eq:relation_internal}) holds because the placement of vertices in $\mathcal{G}^{\text{in}}_i$ is the same for both $\pi'$ and $\pi^*$.
Similarly, Eq. (\ref{eq:relation_external}) follows as $\mathcal{G}^{\text{ex}}_i$ is the same for both $\pi'$ and $\pi$.
Eq. (\ref{eq:relation_cross}) is induced by $\frac{\tau_{u'v'}}{\tau_{u^{*}v^{*}}} \leq \lambda $, where $\lambda = \frac{\max_{i,j\in\mathcal{D}}\tau_{ij}}{\min_{i,j\in\mathcal{D}}\tau_{ij}}$.
Here we use $\tau_{u'v'}$ to represent the traffic cost between the edge servers that locate $u$ and $v$ in $\pi'$, and $\tau_{u^{*}v^{*}}$ for that in $\pi^*$.

Unfolding Eq. (\ref{eq:cost_local_optimum}) yields
\begin{align}
    & C^{+}(\pi|\mathcal{G}^{\text{in}}_i) + C^{+}(\pi|\mathcal{G}^{\text{ex}}_i) +
C^{+}(\pi|\mathcal{G}^{\text{cr}}_i) + C_0(\pi) \notag\\
    \leq \
    & C^{+}(\pi'|\mathcal{G}^{\text{in}}_i) + C^{+}(\pi'|\mathcal{G}^{\text{ex}}_i) +
C^{+}(\pi'|\mathcal{G}^{\text{cr}}_i)  + C_0(\pi').
\end{align}

Substituting Eq. (\ref{eq:relation_internal}), (\ref{eq:relation_external}), and (\ref{eq:relation_cross}), we have
\begin{align}
    & C^{+}(\pi|\mathcal{G}^{\text{in}}_i) + C^{+}(\pi|\mathcal{G}^{\text{ex}}_i) +
C^{+}(\pi|\mathcal{G}^{\text{cr}}_i) + C_0(\pi) \notag\\
    \leq \
    & C^{+}(\pi^*|\mathcal{G}^{\text{in}}_i) + C^{+}(\pi|\mathcal{G}^{\text{ex}}_i) +
\lambda C^{+}(\pi^*|\mathcal{G}^{\text{cr}}_i)  + C_0(\pi'). 
\end{align}

Eliminating redundant items and summarizing all edge servers for the above inequality, we obtain
\begin{align}
    & \sum_{i \in \mathcal{D}} \left[ C^{+}(\pi|\mathcal{G}^{\text{in}}_i) + 
C^{+}(\pi|\mathcal{G}^{\text{cr}}_i) \right]+ C_0(\pi) \notag\\
    \leq \
    & \sum_{i \in \mathcal{D}} \left[ C^{+}(\pi^*|\mathcal{G}^{\text{in}}_i) + \lambda
C^{+}(\pi^*|\mathcal{G}^{\text{cr}}_i) \right] + C_0(\pi'). \label{eq:sum_local_optimum}
\end{align}

For the summarized items in Eq. (\ref{eq:sum_local_optimum}), we have
\begin{align}
    \sum_{i \in \mathcal{D}} C^{+}(\pi|\mathcal{G}^{\text{in}}_i) &= C_1(\pi) + \sum_{i \in \mathcal{D}} C_2(\pi|\mathcal{G}^{\text{in}}_i), \\
    \sum_{i \in \mathcal{D}} C^{+}(\pi|\mathcal{G}^{\text{cr}}_i) &= 2 \sum_{i \in \mathcal{D}} C_2(\pi|\mathcal{G}^{\text{cr}}_i), \label{eq:unfold_pi_cr} \\
    \sum_{i \in \mathcal{D}} C^{+}(\pi^*|\mathcal{G}^{\text{in}}_i) &= C_1(\pi^*) + \sum_{i \in \mathcal{D}} C_2(\pi^*|\mathcal{G}^{\text{in}}_i),\\
    \sum_{i \in \mathcal{D}} C^{+}(\pi^*|\mathcal{G}^{\text{cr}}_i) &= 2 \sum_{i \in \mathcal{D}} C_2(\pi^*|\mathcal{G}^{\text{cr}}_i), \label{eq:unfold_pi*_cr}
\end{align}
where we remark that the right hand in Eq. (\ref{eq:unfold_pi_cr}) and Eq. (\ref{eq:unfold_pi*_cr}) is because every vertices pair $\langle u,v \rangle$ is calculated twice when iterating $i\in\mathcal{D}$. 
Therefore we rewrite Eq. (\ref{eq:sum_local_optimum}) in 
\begin{align}
    & C_1(\pi) + \sum_{i \in \mathcal{D}} C_2(\pi|\mathcal{G}^{\text{in}}_i) + 2 \sum_{i \in \mathcal{D}} C_2(\pi|\mathcal{G}^{\text{cr}}_i) + C_0(\pi) \notag \\
    \leq \
    & C_1(\pi^*) + \sum_{i \in \mathcal{D}} C_2(\pi^*|\mathcal{G}^{\text{in}}_i) + 2\lambda \sum_{i \in \mathcal{D}} C_2(\pi^*|\mathcal{G}^{\text{cr}}_i) + C_0(\pi'). \label{eq:unfold_local_optimum}
\end{align}

Notice that $C(\pi)$ and $C(\pi^*)$ can be alternatively represented as
\begin{align}
    &C(\pi) = C_2(\pi) + C_1(\pi) +  C_0(\pi) \notag \\
    &= \sum_{i \in \mathcal{D}} C_2(\pi|\mathcal{G}^{\text{in}}_i) + \sum_{i \in \mathcal{D}} C_2(\pi|\mathcal{G}^{\text{cr}}_i) + C_1(\pi) +  C_0(\pi), \label{eq:unfold_pi} \\
    &C(\pi^*) = C_2(\pi^*) + C_1(\pi^*) +  C_0(\pi^*) \notag \\
    &= \sum_{i \in \mathcal{D}} C_2(\pi^*|\mathcal{G}^{\text{in}}_i) + \sum_{i \in \mathcal{D}} C_2(\pi^*|\mathcal{G}^{\text{cr}}_i) + C_1(\pi^*) + C_0(\pi^*).   \label{eq:unfold_pi*}
\end{align}

Combining the above equations with Eq. (\ref{eq:unfold_local_optimum}), we have
\begin{align}
    & C(\pi) + \sum_{i \in \mathcal{D}} C_2(\pi|\mathcal{G}^{\text{cr}}_i) - C_0(\pi') \notag \\
    \leq \
    & C(\pi^*) + (2\lambda-1) \sum_{i \in \mathcal{D}} C_2(\pi^*|\mathcal{G}^{\text{cr}}_i) - C_0(\pi^*).
\end{align}

This subsequently leads to
\begin{align}
    C(\pi) 
    \leq \
    & C(\pi^*) + (2\lambda-1) \sum_{i \in \mathcal{D}} C_2(\pi^*|\mathcal{G}^{\text{cr}}_i) \notag \\
    & - \sum_{i \in \mathcal{D}}  C_2(\pi|\mathcal{G}^{\text{cr}}_i) + C_0(\pi') - C_0(\pi^*) \notag \\
    \leq \
    & 2 \lambda C(\pi^*) + C_0(\pi'), \notag \\
    = \
    & 2 \lambda C(\pi^*) + \epsilon.
\end{align}

Theorem \ref{thm:approximation} thereafter holds.
\end{proof}

\section{Proof of Theorem \ref{thm:convergence}}

\begin{proof}
Let $C[\varphi]$ be the obtained system cost by GLAD-S in iteration $\varphi$.
It immediately has $C[\varphi+1] < C[\varphi]$ since GLAD-S continuously seeks a total cost as minimum as possible; otherwise GLAD-S terminates.
Therefore, for all $\xi>0$, there exists a $\varphi$ such that $|C[\varphi]-C'|<\xi$, where $C'$ is within the range of $\left[ C(\pi^*), 2\lambda C(\pi^*) + \epsilon \right]$ when $R = \frac{|\mathcal{D}|(|\mathcal{D}|-1)}{2}$, as given by Theorem \ref{thm:approximation}.
The convergence is thus guaranteed.
\end{proof}

\section{Proof of Theorem \ref{thm:time_complexity}}

\begin{proof}
Given the convergence guaranteed by Theorem \ref{thm:convergence}, we first intend to estimate the upper bound of iteration times $\Phi$.
To accomplish that, we first focus on the number of iterations that successfully updates the graph layout, denoted by $\Phi^{\text{suss}}$.
Assuming $\Delta C$ as the average cost reduction of a successful iteration, we have
\begin{align}
    \Delta C = \frac{C(\pi^0) - C(\pi^*)}{\Phi^{\text{suss}}}, \label{eq:delta_C}
\end{align}
where $C(\pi^0)$ is the system cost of an initial graph layout $\pi^0$ and  $C(\pi^*)$ is the system cost of the converged layout $\pi^*$.
Notice that $\Delta C > 0$ always holds, unless the algorithm cannot decrease system cost anymore and will terminate.
More preciously, we have $\Delta C \geq \delta$, where $\delta$ is the minimum attainable cost reduction induced by the minimum change of an update, \textit{i.e.} switching a vertex from edge $i$ to edge $j$ in a \textit{s-t} cut attempt.
Substituting Eq. (\ref{eq:delta_C}) to the above inequality yields
\begin{align}
    \Phi^{\text{suss}} \leq \frac{C(\pi^0) - C(\pi^*)}{\delta}. \label{eq:phi_with_delta}
\end{align}
In the worst case, the initial layout $\pi^0$ places every vertex to the edge servers that incur the maximum cost and every link between them stretches over different edge servers, which results in a bound of $O(V+E)$ for $C(\pi^0)$.
Provided that $C(\pi^*)>0$, we can induce the following relations from Eq. (\ref{eq:phi_with_delta}):
\begin{align}
    \Phi^{\text{suss}} < \frac{C(\pi^0)}{\delta} \leq \frac{O(V+E)}{\delta}.
\end{align}
Since $\delta$ is constant when parameters $C$, $\pi$, and $\mathcal{G}$ are determined, the upper bound of $\Phi^{\text{suss}}$ is $O(V+E)$.
We next extend $\Phi^{\text{suss}}$ to $\Phi$ by accommodating iterations that fail to update the graph layout.
At worst, the intervals between every two successful iterations always take the most checking times $R$, resulting in an upper bound of $O(R)$.
Therefore, the upper bound of iteration times $\Phi$ is $O((V+E)R)$.

We next analyze the time complexity of each iteration.
Specifically, as we solve the graph cuts using the state-of-the-art max-flow algorithm \cite{Orlin2013max} for implementation, it introduces a time complexity linear to the product of the numbers of involved vertices and links.
In GLAD-S, each $s-t$ cut constructs an auxiliary graph $\mathcal{A}(i,j)$ with maximum vertices number $V+2$ and links number $E+2V$, its solution hence requires $O((V+2)(E+2V))$.
Upon an obtained cut set, to map them into a graph layout $\pi'$ needs to traverse all cuts, with a number linear to the number of vertices $\mathcal{A}(i,j)$, which induces $O(V)$.
Calculating the system cost with a given graph layout is considered a constant time complexity and is omitted.
Stacking the above results yields the time complexity of a single iteration $O((V+2)(E+2V))$, and putting all iterations together produces the time complexity of GLAD-S as $O((V+2)(E+2V)(V+E)R)$.
\end{proof}

\bibliographystyle{IEEEtran}
\bibliography{main.bib}

\end{document}